\documentclass{article}
\usepackage{hyperref}
\hypersetup{colorlinks=true}
\usepackage{latexsym,amsmath, amssymb}
\newcounter{draft} 
\ifnum \thedraft=0
\fi
\makeatletter
\def\subsection{\@startsection {subsection}{2}{\z@}{.8ex
plus 1ex} {1ex}{\sc}}
\makeatother
\newenvironment{proof}{\vspace*{2ex}\noindent {\em Proof:}
}{\hfill $\Box$ \\[2ex]}
\newcommand{\new}{\newcommand}
\newcounter{letter}

\new{\Nabla}{\bigtriangledown}
\new{\Ricci}{\mathrm{Ricci}}
\new{\Pfaff}{\mathrm{Pfaff}}
\new{\Hom}{\mathrm{Hom}}
\new{\iso}{\cong}
\new{\End}{\mathrm{End}}
\new{\maps}{\rightarrow}
\new{\tr}{\mathrm{tr}}
\new{\str}{\mathrm{str}}
\new{\Str}{\mathrm{Str}}
\new{\scalar}{{\mathfrak{r}}}
\new{\defequals}{\stackrel{\mathrm{def}}{=}}
\new{\NN}{\mathbb{N}}
\new{\RR}{\mathbb{R}}
\new{\ZZ}{\mathbb{Z}}
\new{\CC}{\mathbb{C}}
\renewcommand{\vec}{\mathbf}
\new{\gap}{{\;}}
\new{\triplegap}{{\;\;\;\;\;}}
\new{\doublegap}{{\;\;\;}}
\new{\OO}{{\mathcal{O}}}
\new{\PP}{{\mathcal{P}}}
\new{\spinor}{{\mathcal{S}}}
\new{\tspinor}{{\mathcal{T}}}
\newcommand{\Clifford}{{\mathcal{C}}}
\new{\hol}{\operatorname{Hol}}
\new{\dd}{{\mathfrak{d}}}
\new{\dmeas}{{\mathbf{d}}}
\new{\tensor}{\otimes}
\new{\bracket}[1]{\left\langle #1 \right \rangle}
\new{\parens}[1]{\!\left( #1 \right)}
\new{\sbrace}[1]{\!\left[ #1 \right]}
\new{\cbrace}[1]{\!\left\{ #1 \right\}}
\new{\abs}[1]{\left| #1 \r|}
\new{\dirac}{\textsf{\textbf{D}}}
\new{\Vol}{\operatorname{Vol}}
\new{\tnm}[1]{{(#1)}}
\new{\VV}{{\mathcal{V}}}
\new{\XX}{{\mathcal{X}}}
\new{\WW}{\mathcal{W}}
\new{\EE}{{\mathcal{E}}}
\new{\pt}{\mathfrak{P}}
\new{\bint}{\oint}
\new{\pfend}{\noindent \hfill $\Box$}
\new{\Rlimit}{{\mathsf{R}}}
\new{\Flimit}{{\mathsf{F}}}
\new{\beas}{\begin{eqnarray*}}
\new{\eeas}{\end{eqnarray*}}
\renewcommand{\r}{\right}
\renewcommand{\l}{\left}
\new{\norm}[1]{\l|\!\l| #1 \r|\!\r|}
\new{\para}[2]{\subsection{#1. }
\label{par:#2} \ifnum \thedraft=1 \marginpar{\scriptsize{par:#2}} \fi}
\new{\be}{\begin{equation}  
}
\new{\ee}[1]{
\label{eq:#1}
\end{equation} \ifnum \thedraft=1 \marginpar{\scriptsize{\em{eq:#1}}}  
\fi \noindent%
}
\new{\eqa}[2]{\begin{align}
#2 \label{eq:#1}
\end{align}
\ifnum \thedraft=1 \marginpar{\scriptsize{\em{eq:#1}}}  \fi
}
\new{\eqalabel}[1]{
\label{eq:#1}
}
\new{\dlabel}[1]{\ifnum \thedraft=1 \marginpar{\scriptsize{\em{#1}}}  \fi}
\new{\labitem}[1]{%
\label{item:#1} \ifnum \thedraft=1 \marginpar{\scriptsize{\em{item:#1}}}  
\fi \noindent%
}
\numberwithin{equation}{section}
\newtheorem{remark}{Remark}[section]
\newtheorem{theorem}{Theorem}[section]
\newtheorem{proposition}{Proposition}[section]
\newtheorem{lemma}{Lemma}[section]
\newtheorem{corollary}{Corollary}[section]
\newtheorem{definition}{Definition}

\begin{document}

\nocite{Atiyah85} 
\nocite{AJ90}
\nocite{Blau93}
\nocite{DH82} 
\nocite{MQ86}
\nocite{BP08}
 \nocite{Getzler91}
\nocite{BT93}

\ifnum \thedraft=1 \today  \fi

\begin{center}
{\Large Path integrals, SUSY QM  and the Atiyah-Singer index theorem for twisted Dirac}\\
Dana Fine\\
Stephen Sawin\end{center}

\abstract{Feynman's time-slicing construction approximates the path
  integral by a product, determined by a partition of a finite time
  interval, of approximate propagators. This paper formulates general
  conditions to impose on a short-time approximation to the propagator
  in a general class of imaginary-time quantum mechanics on a
  Riemannian manifold which ensure these products converge. The limit
  defines a path integral which agrees pointwise with the heat kernel
  for a generalized Laplacian. The result is a rigorous construction
  of the propagator for supersymmetric quantum mechanics, with
  potential, as a path integral.  Further, the class of Laplacians
  includes the square of the twisted Dirac operator, which corresponds
  to an extension of N=1/2 supersymmetric quantum mechanics. General
  results on the rate of convergence of the approximate path integrals
  suffice in this case to derive the local version of the
  Atiyah-Singer index theorem.}

\section*{Introduction}

This paper's primary goal is to construct imaginary-time path
integrals for a class of theories which includes ordinary quantum
mechanics and what might be called 
twisted $N=1/2$ supersymmetric quantum mechanics (whose precise
definition appears in Sect.~\ref{sc:dirac}) 
on a Riemannian
manifold. The heuristic formulation of such path integrals suggests
they should represent the propagator, i.e.\ the kernel of the
time-evolution operator, which in the imaginary-time formulation is
the heat operator of the given
Laplacian. Further, the
steepest-descent approximation should give asymptotics
for
the heat kernel. Indeed, these two heuristic properties form
the basis for path integral
``proofs'' of index theorems. Therefore this paper constructs the path
integral and then goes on to prove it agrees pointwise  with the kernel
of the heat operator and  to give
an asymptotic approximation in appropriate
circumstances.
This  ensures the construction yields as a
by-product 
new proofs of index
theorems, including the local version of the Atiyah-Singer index
theorem for the twisted Dirac operator. The resulting proof of the latter is
arguably  the closest to the heuristic path-integral argument 
which  Witten \cite{Witten82a,Witten82b} suggests and 
Alvarez-Gaum\'e and Friedan and Windey
\cite{Alvarez83,FW84} implement.
 
The present approach, which is a rigorous realization of Feynman's
 time-slicing interpretation of the path integral, gives approximate
 propagators indexed by partitions of a fixed time interval. These
 approximate propagators can
 be interpreted as defined by an integral over a finite-dimensional
 approximation to the space of paths of a discretized version of the
 action, which in imaginary time is the energy.
 That is, the approximate propagator is a time-slicing
 approximation to the path integral.  The main work of the first two sections is to prove that
 these approximations converge as the partitions get finer.  More
 importantly, the convergence must be sufficiently uniform in the
 parameters defining the theory that steepest descent computes the
  asymptotics of the propagator. The final section checks the earlier
 convergence results suffice to obtain the asymptotics of the
 component of twisted  $N=1/2$ supersymmetric quantum mechanics that
imply the local 
 index theorem.

This paper  is closely related to the authors' previous work
\cite{FS08} constructing the path integral form of the propagator for
imaginary-time $N=1$ supersymmetric quantum mechanics and using it to
prove the Gauss-Bonnet-Chern theorem.  The current work generalizes
the earlier work in two ways.  First, it constructs propagators for a
much larger class of theories (including the $N=1$ version as a
special case, but also including bosonic quantum nechanics and indeed
all theories with elliptic Hamiltonians). Second, to prove
Gauss-Bonnet-Chern  we needed only the lowest order term in the
asymptotics, whereas  the current work  needs the next order in
the asymptotics anf therefore requires a delicate interchange of the
small-parameter and fine-partition limits.  Despite this, the
construction below is somewhat simpler and we hope more natural than
in the earlier work. 

 The authors' above-cited paper discusses the distinctions and
relationships between this approach and other work inspired by the
path integral heuristics including that of
Bismut~\cite{Bismut84a,Bismut84b},
Getzler~\cite{Getzler86a,Getzler86b},
Rogers~\cite{Rogers87,Rogers92a,Rogers92b} and Andersson and
Driver~\cite{AD99}.

  \subsection*{Technical Introduction}

\label{ss:tech}
Secs.~\ref{sc:kernel} and~\ref{sc:limit} of this paper  use Feynman's
time-slicing approach to construct the path integral representing  the
propagator for imaginary-time quantum mechanics on a vector bundle
$\VV$ over a compact 
(or merely ``tame'')
manifold 
$M$ with elliptic quantum Hamiltonian $\Delta$. 
Time slicing starts with an approximate
propagator coming from a discretization of the
action and associates a
product of these kernels to each partition of a given interval of time $[0,t]$.
This leads to a path integral expression for
the exact propagator as a fine-partition  limit. The reference~\cite{FS08}
gives a detailed account of the relation between the time-slicing
approach to defining the path
integral and the refinement limit of a product of approximate
kernels,  with particular attention to the case of $N=1$
supersymmetric quantum mechanics on a Riemannian manifold $M$. For a
look at how this works in a simple case, consider a (bosonic) Lagrangian
\newcommand{\sigmadot}{\dot{\sigma}}
$L(\sigma, \sigmadot, s)$ depending on a  path $\sigma : [0,t]
\rightarrow M$, with parameter $s$ and tangent
$\sigmadot$. Heuristically, the kernel of the
time-evolution operator $e^{-t \Delta/2}$ may be written as the path
integral (with imaginary time and $\hbar=1$ units)
\[
\int e^{-\int_0^t L \, ds} \dd \sigma,
\] where the integral is over paths with $\sigma(0)=y$ and
$\sigma(t)=x$.  Note the endpoint conditions and the explicit
$t$-dependence mean the path integral is a function on $M \times
M \times R$, as is the propagator. In the imaginary-time 
formulation, the time-evolution operator is in fact the heat operator
associated with the Hamiltonian $\Delta$, which is a generalized
Laplacian.

The idea 
of time-slicing is to partition $[0,t]$ into subintervals of length
$t_i$ for $i = 1,2 \ldots n$, and to write the path integral as a product of $n$
 such
integrals. Then, in each of these path integrals,  replace the integral of $L$ over the subinterval of length
$t_i$ with an approximation $\widehat{L}(y_i, y_{i-1}; t_i) t_i$, where
the $y$'s are the endpoint values of $\sigma$ on that
subinterval. Heuristically, a requirement on this approximation  is
that $ \sum\widehat{L}(y_i, y_{i-1}; t_i) t_i$ be a Riemann
sum converging under refinement to $\int_0^t L \, ds$. This leads to
an approximate heat kernel $K(x,y;t)= (2 \pi t)^{-m/2}e^{-\widehat{L}(x,y;t)t}$, and a well-defined
approximate path integral which is the kernel product of $n$ copies of
$K$. The Riemann sum requirement suggests that if $t$ itself is small enough, the
trivial partition should suffice; hence, $K$ must be close
to the actual heat kernel when $t$ is small. If $K$ has the semigroup
property, then in fact the approximation is independent of the choice
of partition, and the convergence of the approximate path integral is
trivial. 

One obvious choice for $\widehat{L}$ is to ask that
$\widehat{L}(x,y;t) t = \int_0^t
L(\sigma_{\mathrm{cl}},\sigmadot_{\mathrm{cl}};s) \, ds$ where
$\sigma_{\mathrm{cl}}$ is the path obeying the classical equations of
motion subject to $\sigma_{\mathrm{cl}}(0)= y$ and
$\sigma_{\mathrm{cl}}(t)= x$.  Suprisingly, this does not lead to a
limit with the desired Hamiltonian;  correction terms, which may be
thought of as resolving  operator-ordering ambiguity, must be
added. (In physical units, these corrections enter at higher powers of
$\hbar$.)

Sec.~\ref{sc:kernel} spells out, in a general setting, an appropriate
sense of $K$ being almost the heat kernel for a given choice of
$\Delta,$ and provides the needed estimates.  In particular
Def.~\ref{def:approx-semigroup} 
spells out how close $K$ must be to satisfying the semigroup property
to ensure the kernel products defining the path integrals converge,
and Def~\ref{def:ahk} says how close $K$ must be to the true heat
kernel of a given $\Delta$ to ensure the limit is the heat kernel.
Sec.~\ref{sc:limit} proves the existence of the 
fine-partition limit for such $K$, as well as properties of the
limiting kernel, and precise results on the convergence.
Sec.~\ref{sc:dirac} associates a quantum mechanical system to each
generalized Laplacian, by relating a given action to a path integral
construction for the corresponding propagator. In particular, it gives
the propagator for twisted $N=1/2$ supersymmetric quantum mechanics.
Sec.~\ref{sc:rescale} treats the asymptotic behavior of this
propagator, which requires using results from Sec.~\ref{sc:limit} to
interchange the asymptotic and fine-partition limits. The result agrees
with the heuristic steepest descent treatment of the path integral.


\section{Approximate heat kernels} \label{sc:kernel}

\subsection{Kernels, $*$-products,  and local coordinate bounds} \label{ss:coords}

The heuristic time-slicing interpretation of the path integral
suggests, as above, the approximation to the heat kernel need only get
the short-time and near-diagonal (on $M \times M$) behavior
right. This suggests formulating the
requirements on an approximation locally. 

Accordingly, let $O$ be 
 an open  contractible
subset of
$\RR^m,$ and let $g_{ij}(x)$ be a smooth Riemannian metric on
$O$. Require  that  all derivatives
of order $k$ of $g$ and of $g^{-1}$ are bounded in supremum norm for
$0 \leq k \leq 5.$ This will ultimately ensure given approximations to
the short-time behavior have the desired convergence properties.

Let $d(x,y)$ be the distance between $x,y\in O$ in this metric.  For
$\vec{v} \in \RR^m,$ $x \in O$ and $t \in \RR$ the 
geodesic  through
$x$ with tangent $\vec{v}$ at $x$ with parameter $t$ proportional to
arc length defines the exponential map $\exp_x t \vec{v}$. If $y\in O$ is close enough to $x$ that
there is a unique minimal geodesic connecting them, define
$\vec{y}_x=\exp^{-1}_x y.$ Let $\parens{\,\cdot,\cdot\,}_x$
denote the
inner product with respect to $g$ at $x \in O,$ and let
$\abs{\,\cdot\,}_x$ denote the corresponding norm. If the vectors
inside are of the form $\vec{y}_x$ or the point at which the norm or
inner product is computed is otherwise understood from context, drop
the subscript.  Write $\dmeas_g y= \det^{1/2}_y(g)\dmeas y,$ where
$\dmeas y$ is standard Lebesgue measure on $\RR^m$ restricted to $O$, and write $\dmeas
\vec{y}_x$ for Lebesgue measure on $O$ with respect to the inner
product given by $g$ at $x;$ that is, the metric measure at $x$ pulled
back to $y$ by $\exp^{-1}_x$.

Henceforth to say that a quantity, such as $D$ in the following lemma,
``depends on the metric bounds'' will mean that quantity is a
function of the assumed bounds on the supremum norm of
 $g$, $g^{-1}$
and their first five derivatives (as well as on the dimension
$m$)
 The concern is that, in later
arguments which require rescaling the metric, preserving these bounds
should be
sufficient to preserve the estimates which follow here.

\begin{lemma}
	\label{lm:coords}
	There is a $D>0$ depending on the metric bounds  such that .
        for $x,y,z\in O$ with $d(x,y),d(y,z),d(x,z)<D$  there is a unique minimal
        geodesic connecting $x$ and $y$,  $\vec{y}_x$
	depends smoothly on $x$ and $y$,  
        and $y-x$ depends smoothly on $x$ and on $\vec{y}_x$.
Moreover, 
	\begin{align}
          y-x&= \vec{y}_x + \OO\parens{\abs{\vec{y}_x}^2} \label{eq:yminusx-est}\\
          \abs{\vec{z}_x}^2 &= \abs{\vec{z}_y}^2 +\abs{\vec{x}_y}^2 -
          2 \parens{\vec{x}_y,\vec{z}_y} +
          \OO\parens{\abs{\vec{x}_y}^2
            \abs{\vec{z}_y}^2}\label{eq:lengthsquared}\\ 
         \frac{\dmeas_gy}{\dmeas \vec{y}_x}&=1+\OO\parens{\abs{\vec{y}_x}^2} \label{eq:determinant}
	\end{align}
	where for example $\OO\parens{\abs{\vec{x}_y}^2
          \abs{\vec{z}_y}^2}$ 
        indicates the difference between the left-hand side and the
        truncated Taylor series is bounded by a constant (depending on
        the metric bounds) times 
        $\abs{\vec{x}_y}^2 \abs{\vec{z}_y}^2$
        (as each of these tends towards zero).
\end{lemma}

\begin{proof}
As expressed in local coordinates, the components of the Riemann
curvature  are continuous 
functions of the first
two derivatives of the metric. By assumption, then, the 
Riemannian and hence sectional
curvatures are bounded above, so, by Rauch's comparison
theorem~\cite{doCarmo92}, the injectivity radius is bounded
below (contractibility means the injectivity radius is the minimum
distance of a point from its nearest conjugate point).
 Within the injectivity radius the exponential map $\exp_x$ at each
 $x$ is defined 
by the differential equation in local coordinates, writing
$\sigma^\mu(t)$ for the $\mu$th component of $\exp_x(t\vec{v}),$
\[
\frac{d^2 \sigma^\mu}{dt^2} + \Gamma^{\mu}_{\nu \rho} \frac{d
  \sigma^\nu}{dt} \frac{d \sigma^\rho}{dt} = 0.
\]
Since the Christoffel symbols $\Gamma^{\mu}_{\nu \rho}$
are continuous  in the first
derivatives of the metric, the coefficients of the differential
equation have bounded derivatives up to degree four.  Standard
existence and uniqueness results \cite{Arnold98} ensure the solution
with the given initial conditions is $C^5$, but a careful reading of
the argument shows that the first four derivatives are in fact bounded
in terms of the metric bounds.  Further, with $t=1$, $\exp_x \vec{v}$
has its first four derivatives with respect to both $x$ and $\vec{v}$
bounded in terms of the metric bounds. (One normally thinks of $\exp$
as a map from the tangent space to the manifold, but in this case each
of these is a subset of $\RR^m$, so $\exp$ refers to the endomorphism
on $\RR^m$). Since the injectivity radius is bounded below by the
metric bounds, there is a radius $D$ bounded below by the metric
bounds such that $\exp_x^{-1}$ has its first four derivatives bounded
in terms of the metric bounds on a circle of radius $D$ around $x$.

This means that if $d(x,y)<D$ then $\vec{y}_x=\exp^{-1}_xy$ as a
function of $x$ and $y$ has its first four derivatives bounded in
terms of the metric bounds  and  $y= \exp_x \vec{y}_x$
and hence $y-x$ as functions of $x$ and $\vec{y}_x$ have their first
four derivatives bounded in terms of the metric bounds.

For Eq.~\eqref{eq:yminusx-est}  the zeroth and first order terms of
the Taylor series for $y=\exp_x \vec{y}_x$ as a
function of $x$ are set by the initial conditions of $\exp,$ and the
second order error term is bounded by the supremum of the second
derivative of $\exp,$ which is bounded in terms of the metric bounds.

For Eq.~\eqref{eq:lengthsquared}, fixing $y,$ notice that
$\abs{\vec{z}_x}= d(x,z)= d(\exp_y \vec{x}_y , \exp_y \vec{z}_y)$ and all its first
four derivatives in $\vec{x}_y$ and $\vec{z}_y$ are bounded in terms
of the metric bounds.  With Gauss' Lemma,  the Taylor series of
$d(x,z)^2$ as a function of $\vec{x}_y$ is 
\[d(x,z)^2 = \abs{\vec{z}_y}^2 - 2 \parens{\vec{x}_y, \vec{z}_y} +
\vec{x}_y^2 \frac{\partial^2}{\partial \vec{x}_y^2}d(x',z)^2\]
the last term on the
right-hand side is an abbreviation for a linear combination of
quadratic functions of $\vec{x}_y$ involving  second partial derivatives with respect to the components
of $\vec{x}_y$, each evaluated at some point
  $x'$ on the geodesic from $y$ to
$x.$ (The point $x'$ will in general be different for each of the
derivatives appearing in the linear combination.) 

  Expanding this last term
term as a Taylor series in $\vec{z}_y$ yields
\[ d(x,z)^2 = \abs{\vec{x}_y}^2 + \abs{\vec{z}_y}^2 - 2 \parens{\vec{x}_y, \vec{z}_y} +
\vec{x}_y^2\vec{z}_y^2\frac{\partial^2 }{\partial \vec{z}_y^2}\frac{\partial^2}{\partial
\vec{x}_y^2}d(x',z')^2\]
where $z'$ is on the geodesic between $y$ and $z.$  The last term on
the right-hand side here, being a fourth derivative of the exponential
map, is bounded in terms of the
metric bounds.

For Eq.~\eqref{eq:determinant}, note that the Taylor series centered
at $x$ for the
components of the metric at $y,$ expressed in the 
coordinates mapping $y$ to $\vec{y}_x$, will have no term linear in
$\vec{y}_x$; the quadratic term has
coefficients given by second derivatives of the metric at
$x$  \cite{BGV04}. Eq.~\eqref{eq:determinant} follows by direct calculation, with
the implied constants depending on the bounds of the
metric.
\end{proof}


Define
\[
\chi_{<D}(x,y) = \begin{cases}
	1& \quad \text{if} \quad d(x,y)<D\\
	0 &	\quad \text{else,}
\end{cases}
\]
and  $\chi_{>D}(x,y)=1-\chi_{<D}(x,y).$  

For $x,y\in O,$ $t>0,$ and $D>0$ small enough that
Lemma~\ref{lm:coords} holds, define
\be H_D(x,y;t)= \chi_{<D}(x,y)
  (2\pi t)^{-m/2} e^{-\abs{\vec{y}_x}^2/(2t)}. 
\ee{h-def}

Given $n \in \NN,$ let $f\colon O \to \RR^n,$ $f^* \colon O
\to \parens{\RR^n}^*$ and $K \colon O \times O \to \operatorname{Matrix}_{n,n}.$ $K$
represents kernels of left or right operators on the space of
functions from $O$ to $\RR^n$ or $\parens{\RR^n}^*$
whose actions are
given by
\eqa{*-def}{
	K*f(x)&= \int_O K(x,y) \cdot f(y) \dmeas_gy \nonumber\\
	f^**K(y)&= \int_O f^*(x) \cdot K(x,y) \dmeas_gx 
}
where $\cdot$
represents the matrix product. The kernel of the operator product of
the operators represented by $K$ and $J$ is the *-product 
\eqa{*-product}{
J*K(x,z)&= \int_O J(x,y) \cdot K(y,z) \dmeas_gy.
}

The matrix norm 
sends
 $K$ to 
a nonnegative function  $\abs{K}$ on $O \times O.$   
Use this to define 
\[ \norm{K}_{\mathrm{op}}= \max\parens{\sup_x \int \abs{K(x,y)} \dmeas_g y, \sup_y \int  \abs{K(x,y)}\dmeas_g x},\]
 which is the max of the operator norms of $K$ acting on the left and
 the right. 
 Define the kernel norm by \[
 \norm{K}_\mathrm{ker}=\max(\norm{K}_\mathrm{op},\norm{K}_\infty). \]
 Notice $\norm{J*K}_\mathrm{ker} \leq \norm{J}_\mathrm{ker}\norm{K}_\mathrm{ker}$ and $\norm{J*K}_\mathrm{ker} \leq
 \norm{J}_\mathrm{op}\norm{K}_\mathrm{ker}$.

Notice $H_D$ of Eq.~\eqref{eq:h-def} agrees for $d(x,y) < D$ with the
flat-space heat kernel when the metric is flat. 
The next two lemmas explore classes of kernels
whose relation to $H_D$ are increasingly tenuous, to delineate the
extent to which they retain key properties of the heat kernel under
kernel products. The purpose of this exploration, which culminates in
Prop.~\ref{pr:t-def}, is to determine the key properties of a
time-slicing approximation that ensure the approximate path integrals
converge with sufficient rapidity to the heat kernel of a given
Laplace-like operator.

 \begin{lemma}
 	\label{lm:h*h-bd}
 	If $B$ is large enough,  $D$ is small enough, and $t$ is small
        enough (each depending 
        on the bounds of the metric and the previous quantities); and if
        \be
        K_{B,D}(x,y;t)=e^{B\abs{\vec{y}_x}^2/(5m)}H_{D}(x,y;t),
        \ee{kbd-def}
        $0<t_1,t_2,$ and $t=t_1+t_2;$ then 
 	\eqa{h*h-bd}{
 	  \chi_{<D}\sbrace{K_{B,D}(t_1)*K_{B,D}(t_2)} &\leq  e^{Bt_1t_2/t}K_{B,D}(t),\nonumber \\
 	\norm{\chi_{>D}\sbrace{K_{B,D}(t_1)*K_{B,D}(t_2)}}_{\mathrm{ker}} & \leq t^2e^{-D^2/9t},
 	}
 	and
 	\be
 	\norm{K_{B,D}(t)}_{\mathrm{op}}\leq e^{Bt}.
 	\ee{hopnorm-bd}
 \end{lemma}

 \begin{proof}
   For the first line of Eq.~\eqref{eq:h*h-bd}, let
   $u=\exp_x \parens{\frac{t_1}{t}\vec{z}_x},$ so $\vec{u}_x = t_1
   \vec{z}_x/t$ and $\vec{u}_z = t_2 \vec{x}_z/t.$
   These imply
   $t_2\vec{x_u} + t_1 \vec{z_u}=0.$ Let $b=B/(5m)$ and let $c>0$ be
   such that Eqs.~\eqref{eq:lengthsquared} and~\eqref{eq:determinant}
   become
	\begin{align*}
		\abs{\abs{\vec{y}_x}^2 - \parens{\abs{\vec{y}_u}^2 +
                    \abs{\vec{x}_u}^2  - 2
                    \parens{\vec{y}_u,\vec{x}_u}} }&\leq c\abs{\vec{y}_u}^2\abs{\vec{x}_u}^2, \\
		\abs{\abs{\vec{y}_z}^2 -\parens{\abs{\vec{y}_u}^2 +
                    \abs{\vec{z}_u}^2  - 2
                    \parens{\vec{y}_u,\vec{z}_u}} }&\leq
                c\abs{\vec{y}_u}^2\abs{\vec{z}_u}^2, \mbox{ and } \\
		\dmeas_gy &\leq \sbrace{1+ c\abs{\vec{y}_u}^2}\dmeas\vec{y}_u.
	\end{align*}
	Then 
	\begin{align*}
		&\chi_{<D}(x,z)\sbrace{K_{B,D}(t_1)*K_{B,D}(t_2)}(x,z) = \int
                \chi_{<D}(x,z)K_{B,D}(x,y;t_1)K_{B,D}(y,z;t_2) \dmeas_gy\\ 
		& \qquad \leq H_{D}(x,z;t) e^{b(t_1^2+t_2^2)\abs{\vec{z}_x}^2/t^2} (2 \pi t_1t_2/t)^{-m/2}\int \chi_{<D}(x,y)\chi_{<D}(y,z)\\
		& \qquad \qquad 
                \cdot e^{-\frac{t\abs{\vec{y}_u}^2}{2t_1t_2}\sbrace{1-ct_1t_2\abs{\vec{z}_x}^2/t^2
                    - 2ct_1t_2/t - 4b t_1t_2/t -2bc t_1t_2
                    \abs{\vec{z}_x}^2/t} -
                  2b\parens{\vec{y}_u,\vec{x}_u+\vec{z_u}}} \dmeas
                \vec{y}_u 
	\end{align*}
	where the last integral is over all vectors $\vec{y}_u$ which
        are 
taken by
 $\exp_u$ to some $y$ within a
        distance $D$ of $x$ and $z.$ Because the integrand is positive
        the inequality still holds if the integral is extended over
        all of $\RR^m.$ Noting that $e^x \leq x + e^{x^2}$ and that
        the integral of a Gaussian times a linear function is $0,$
	\[
\chi_{<D} \sbrace{K_{B,D}(t_1)*K_{B,D}(t_2)}(x,z)  \leq H_{D}(x,z;t)
e^{b(t_1^2+t_2^2)\abs{\vec{z}_x}^2/t^2} \int_{\RR^m} (2\pi
t_1t_2/t)^{-m/2} e^{-\frac{t\abs{\vec{y}_u}^2}{2t_1t_2}\sbrace{1-a}}
\dmeas\vec{y}_u\] 
	where
	\begin{align*} a&= ct_1t_2\abs{\vec{z}_x}^2/t^2 + 2ct_1t_2/t +
          4b t_1t_2/t + 2b(c + 4b) t_1t_2 \abs{\vec{z}_x}^2/t\\ 
		&= (2c+4b)t_1t_2/t + (c+
                2b(c+4b)t)t_1t_2\abs{\vec{z}_x}^2/t^2 \leq
                Bt_1t_2/(mt) + 2bt_1t_2\abs{\vec{z}_x}^2/(mt^2) 
	\end{align*}
	if $B$ is chosen large enough and $T$ is chosen small enough.  
	If $t$ and $D$ are small enough (depending on $B$)
 then $a\leq 1/2,$ so 
$(1-a)^{-m/2} < e^{ma}$ and hence the Gaussian integral yields 

	\[\chi_{<D}\sbrace{K_{B,D}(t_1)*K_{B,D}(t_2)}(x,z) \leq H_{D}(x,z)
        e^{b\abs{\vec{z}_x}^2+ Bt_1t_2/t}.\]

        Deferring the proof of the second line of Eq.~\eqref{eq:h*h-bd} for a moment,
        consider first Eq.~\eqref{eq:hopnorm-bd}: Defining $b$ and $c$ as above
	\begin{align*}
		\abs{K_{B,D}(t)*f(x)} & \leq  \int H_D(x,y;t) e^{b \abs{\vec{y}_x}^2 }  \abs{f(y)}\dmeas_g y\\
		&\leq \norm{f}_\infty\int \chi_{<D}(x,y)(2\pi
                t)^{-m/2}
                e^{-\frac{\abs{\vec{y}_x}^2}{2t}\sbrace{1-2bt
                    -2ct}}\dmeas \vec{y}_x. 
	\end{align*}
	Again extending the integral, choosing $t$ small enough to
        bound the quantity in braces, and completing the Gaussian
        integral yields
	\[\abs{K_{B,D}(t)*f(x)} \leq \norm{f}_\infty e^{2m(b+c)t}\leq e^{Bt}\norm{f}_\infty\]
	if $B$ is chosen large enough.
	
	Finally, for the second line of Eq.~\eqref{eq:h*h-bd}, if $d(x,z)>D$
        then any $y \in O$ satisfies either $d(x,y)>D/2$ or $d(y,z)>D/2,$ so
	\begin{align*}
		\chi_{>D}(x,z)\sbrace{K_{B,D}(t_1)*K_{B,D}(t_2)}(x,z) &\leq \int \chi_{>D/2}(x,y)K_{B,D}(x,y;t_1)K_{B,D}(y,z;t_2) \dmeas_gy \\
		& \qquad + \int\chi_{>D/2}(y,z)K_{B,D}(x,y;t_1)K_{B,D}(y,z;t_2) \dmeas_gy.
		\end{align*}
	But if $d(x,y)>D/2$, then $K_{B,D}(x,y;t_1) \leq
       (1/2) t_1^2e^{-D^2/9t_1}$ 
        if $t$ is  small enough, and, by
        Eq.~\eqref{eq:hopnorm-bd}, $\norm{K_{B,D}(t_2)}_\mathrm{op} \leq 2$
        for $T$ small enough. Thus 
	\[\int\chi_{>D/2}(x,y) K_{B,D}(x,y;t_1)K_{B,D}(y,z;t_2) \dmeas_g y \leq t_1^2e^{-D^2/(9t_1)},\]
	and therefore 
	\[\norm{\chi_{>D}\sbrace{K_{B,D}(t_1)*K_{B,D}(t_2)}(x,z) }_\infty \leq
        t_1^2e^{-D^2/(9t_1)}+ t_2^2e^{-D^2/(9t_2)} \leq t^2
        e^{-D^2/(9t)}\] 
	 by the convexity of $e^{-D^2/(9t)},$ all for small enough $t$
         (depending on $D$).
	
         If $D$ is chosen small enough, the volume of the ball of
         radius $2D$ around any point is less than $1$ (based on the
         bound on the second derivative of $g$) so
	\[\norm{\chi_{>D}\sbrace{K_{B,D}(t_1)*K_{B,D}(t_2)}(x,z) }_\mathrm{op}=
        \norm{\chi_{>D}\sbrace{K_{B,D}(t_1)*K_{B,D}(t_2)}(x,z) }_\mathrm{ker} \leq
        t^2 e^{-D^2/(9t)}.\] 
 \end{proof}

\subsection{Two families of kernels and the $t$-norm}
\begin{definition} \label{def:EE}
For $B,D,t>0$ define $\EE_{B,D}(t)$  to
be the set of all kernels $K$ 
for which
there exists a probability measure $\dmeas\mu$ on the interval $[1,2]$
such that 
\be
         \abs{K(x,y)}  \leq e^{B\sqrt{t}}\int K_{B,D}(x,y;\alpha t)
        \dmeas \mu_\alpha,
\ee{ebd-def}
where $K_{B,D}$ is the particular one-parameter family of kernels defined in
Eq.~\eqref{eq:kbd-def}.
\end{definition}
Note that $K_{B,D}(t)$ 
itself is in
$\EE_{B,D}(t)$. The following lemma extends the previous one to say
that $\EE_{B,D}(t)$ is almost closed under the $*$ product, and made up
of almost  contraction maps.
The ``almost'' here refers in both cases
to the exponential $\sqrt{t}$ factor, and in the first  to an
exponentially damped term far from the diagonal. Precisely, 
\begin{lemma}
	\label{lm:k*k-bd}
	If $B$ is large enough, $D$ is small enough, and $T$ is small
        enough (each depending  on
        the bounds of the metric and the previous quantities) and if
        $K_1$ and $K_2$ are one-parameter families of kernels with
        $K_1(t),K_2 (t) \in \EE_{B,D}(t)$  for $t < T,$ then, for
        $0<t_1,t_2$ and $t=t_1+t_2 < T$ 
	\be
	\norm{K_i(t)}_{\mathrm{op}}\leq e^{1.1 B\sqrt{t}}
	\ee{kopnorm-bd}
	and
	\eqa{k*k-bd}{
	  \chi_{<D}K_1(t_1)*K_2(t_2) &\in e^{B\sqrt{t_1t_2/t}}\EE_{B,D}(t)\nonumber \\
	\norm{\chi_{>D}K_1(t_1)*K_2(t_2)}_{\mathrm{ker}} & \leq t^2e^{-D^2/(20t)}.
	}

\end{lemma}

\begin{proof}For $i = 1,2$, choose measures $\dmeas \mu_{i,\alpha}$ on
  $[1,2]$ such that $\abs{K_i(t)}\leq
  e^{B\sqrt{t}}\int_1^2K_{B,D}(\alpha t) \dmeas \mu_{i,\alpha}$.  For
  Eq.~\eqref{eq:kopnorm-bd}, 
	\begin{align*}
		\norm{K_1(t)}_{\mathrm{op}} & \leq
                e^{B\sqrt{t}}\norm{\int_1^2 K_{B,D}(\alpha
                  t)\dmeas\mu_{1,\alpha}}_{\mathrm{op}}\\ 
		& \leq e^{B\sqrt{t}}\int_1^2 e^{B \alpha t}
                \dmeas\mu_{1,\alpha} \leq e^{B \sqrt{t} + 2Bt} \leq
                e^{1.1B\sqrt{t}}  ,
	\end{align*}
using	Eq.~\eqref{eq:hopnorm-bd}, and assuming  $t$ is  small enough
for the final inequality to hold.
 For the first line of
  Eq.~\eqref{eq:k*k-bd}, use the first line of Eq.~\eqref{eq:h*h-bd}
  to get
	\begin{align*}
		\abs{\chi_{<D}\sbrace{K_1(t_1)*K_2(t_2)}} &\leq
                e^{B\sqrt{t_1}+B\sqrt{t_2}}\int_1^2\int_1^2\chi_{<D}K_{B,D}(\alpha
                t_1)*K_{B,D}(\beta t_2) \dmeas\mu_{1,\alpha}
                \dmeas\mu_{2,\beta}\\ 
		&\leq
                e^{B\sqrt{t_1}+B\sqrt{t_2}}\int_1^2\int_1^2e^{B\alpha
                  \beta t_1 t_2/(\alpha t_1+ \beta t_2)} K_{B,D}(\alpha
                t_1+\beta t_2) \dmeas\mu_{1,\alpha}
                \dmeas\mu_{2,\beta}\\ 
		&\leq e^{B\sqrt{t_1}+B\sqrt{t_2} + 2Bt_1t_2/t}\int_1^2 K_{B,D}(\gamma t) \dmeas\nu_{\gamma} \\
		&\leq e^{B\sqrt{t}}e^{B\sqrt{t_1t_2/t}}\int_1^2 K_{B,D}(\gamma t) \dmeas\nu_{\gamma}  \in e^{B\sqrt{t_1t_2/t}}\EE_{B,D}(t),
	\end{align*}
	if $t$ is  small enough. Here $\gamma t= \alpha t_1 +
        \beta t_2$ and $\dmeas \nu$ is the pushforward of the product
        measure $\dmeas\mu_1 \dmeas\mu_2$ to this
        subspace.
	
	For the second line of Eq.~\eqref{eq:k*k-bd}, 
	\begin{align*}
		\norm{\chi_{>D}\sbrace{K_1(t_1)*K_2(t_2)}}_{\mathrm{ker}}
                &\leq \norm{e^{B\sqrt{t_1}+ B
                    \sqrt{t_2}}\int_1^2\int_1^2 \chi_{>D}
                  \sbrace{K_{B,D}(\alpha t_1)*K_{B,D}(\beta t_2)}
                  \dmeas\mu_{1,\alpha}
                  \dmeas\mu_{2,\beta}}_{\mathrm{ker}} \\ 
		& \leq e^{2B\sqrt{t}}\int_1^2\int_1^2\norm{\chi_{>D}
                  \sbrace{K_{B,D}(\alpha t_1)*K_{B,D}(\beta t_2)}}_{\mathrm{ker}}
                \dmeas\mu_{1,\alpha} \dmeas\mu_{2,\beta} \\ 
		&\leq e^{2B\sqrt{t}}\int_1^2\int_1^2 4 t^2e^{-D^2/(18t)} \dmeas\mu_{1,\alpha} \dmeas\mu_{2,\beta}\\
		&\leq t^2 e^{-D^2/(20t)}
	\end{align*}
	if $t$ is  small enough. The third line here follows
        from  Eq.~\eqref{eq:h*h-bd}.
\end{proof}

Continue to enlarge the class of kernels which behave  well
under kernel products to
\begin{definition} \label{def:EE'}
For, $B,D,t>0$ define $\EE'_{B,D}(t)$ to be the set
of all kernels which can be written as $K  + J$ where $K \in
\EE_{B,D}(t)$ and $\norm{J}_\mathrm{ker}\leq
te^{-D^2/(20t)}$. 
\end{definition}
This
class is also almost closed under kernel products, in a sense which
the following proposition makes precise. 

\begin{proposition}
	\label{pr:t-def}
	If $B$ is large enough,  $D$ is small enough and $T$ is small
        enough (each depending only
        on the bounds of the metric and the previous quantities)
        and if 
        $K_1$ and $K_2$ are one-parameter families of kernels with
        $K_1(t),K_2 (t) \in \EE'_{B,D}(t)$  for all $t < T,$ then, for
        $0<t_1,t_2$ and $t=t_1+t_2 < T$ 
 	\be
	\norm{K_i(t)}_{\mathrm{op}}\leq e^{2B\sqrt{t}},
	\ee{kop-bd} 
 
	\be
	\abs{K_i(x,y;t)} \leq 2 (2 \pi t)^{-m/2} e^{-d(x,y)^2/(4t)} + t e^{-D^2/(20t)},
	\ee{inf-bd}
    and    	
        	\be
	  K_1(t_1)*K_2(t_2) \in e^{B\sqrt{t}}\EE'_{B,D}(t).
	\ee{*-bd}
\end{proposition}

\begin{remark}
	There is a minimum $B$ and a maximum $D$ and $T$  to make
        Prop.~\ref{pr:t-def}
        hold, and these numbers depend only
        on the supremum of the first few derivatives of the metric and
        its inverse (and $m$), a fact that will be crucial in
        Sect.~\ref{sc:rescale}.
        If  one chose a larger $B,$ the
        maximum $D$ and $T$ would be smaller but would still exist.
        If one chose an even smaller $D,$ the maximum $T$ would be
        smaller still.   In the definition of approximate
        semigroup and approximate heat kernel below, the choice of
        constants will also depend on the family of kernels being
        considered.
	\end{remark}

\begin{proof}
  	For Eq.~\eqref{eq:kop-bd}, write $K_i(t)= \widetilde{K}_i(t) + J_i(t)$  where $\widetilde{K}_i(t)\in \EE_{B,D}(t)$
  and $\norm{J_i(t)}_{\mathrm{ker}}\leq t e^{-D^2/(20t)}.$ Then, using Eq.~\eqref{eq:kopnorm-bd},
\[
	\norm{K_i(t)}_{\mathrm{op}}\leq
        \norm{\widetilde{K}_i(t)}_{\mathrm{op}} +
        \norm{J_i(t)}_{\mathrm{op}}\leq e^{1.1B\sqrt{t}} + e^{-D^2/t}
        \leq e^{2B\sqrt{t}}
\]
for small enough $t.$ 	Eq.~\eqref{eq:inf-bd} follows from the
definition of $\EE'_{B,D}$ for small enough $D$ and $t.$

 For Eq.~\eqref{eq:*-bd}, use Eq.~\eqref{eq:k*k-bd} to write
  $\widetilde{K}_1(t_1)*\widetilde{K}_2(t_2)=\widetilde{K}_3(t_1,t_2) + J_3(t_1,t_2),$ where
  $\widetilde{K}_3(t_1,t_2) \in e^{B\sqrt{t}}\EE_{B,D}(t)$ and
  $\norm{J_3(t_1,t_2)}_{\mathrm{ker}} \leq
  t^2e^{-D^2/(20t)}.$
Then
	\begin{align*}
          &\norm{K_1(t_1)*K_2(t_2) - \widetilde{K}_3(t_1,t_2)}_{\mathrm{ker}} \\
          &\qquad \leq  \norm{J_3(t_1,t_2)}_{\mathrm{ker}} +
          \norm{\widetilde{K}_1(t_1)*J_2(t_2)}_{\mathrm{ker}} +
          \norm{J_1(t_1)*\widetilde{K}_2(t_2)}_{\mathrm{ker}} +
          \norm{J_1(t_1)*J_2(t_2)}_{\mathrm{ker}}\\ 
          &\qquad \leq  t^2e^{-D^2/(20t)} +
          \norm{\widetilde{K}_1(t_1)}_{\mathrm{op}}\norm{J_2(t_2)}_{\mathrm{ker}}
          +
          \norm{J_1(t)1)}_{\mathrm{ker}}\norm{\widetilde{K}_2(t_2)}_{\mathrm{op}}
          +
          \norm{J_1(t_1)}_{\mathrm{ker}}\norm{J_2(t_2)}_{\mathrm{ker}}\\ 
          &\qquad \leq  t^2e^{-D^2/(20t)} +
          e^{1.1Bt_1^{1/2}}t_2e^{-D^2/(20t_2)} +
          t_1e^{-D^2/(20t_1)}e^{1.1Bt_2^{1/2}} +
          t_1e^{-D^2/(20t_1)}t_2e^{-D^2/(20t_2)}\\ 
          &\qquad \leq 
         \frac{5}{4}t^2 e^{-D^2/(20t)} +  te^{1.1B(t/2)^{1/2}}e^{-D^2/(20t)} \leq e^{Bt^{1/2}}t e^{-D^2/(20t)}
	\end{align*} 
	where the third inequality comes from
        Eq.~\eqref{eq:kopnorm-bd}. 
        In the fourth, one straightforward estimate gives
        $t_1e^{-D^2/(20t_1)}t_2e^{-D^2/(20t_2)} \leq \frac{1}{4}
        t^2e^{-D^2/(20t)}$. Further, $t_2e^{1.1Bt_1^{1/2}} +
        t_1e^{1.1Bt_2^{1/2}} \leq te^{1.1B(t/2)^{1/2}}$, for small
        enough $t$.
        The fifth is a straightforward estimate for
        small enough $t.$

\end{proof}
This proposition provides the basis on which to
define a norm:
\begin{definition} \label{def:t-norm}
For given $B, D, t > 0$  define the $t$-norm $\norm{K}_{(t)}$ to be the
smallest positive real number such that $K/\norm{K}_\tnm{t} \in
\EE'_{B,D}(t)$ if it exists. (Otherwise set $\norm{K}_\tnm{t} =
\infty$.) 
\end{definition}
\begin{corollary}
  \label{cr:t-norm}
  If $B$ is large enough,  $D$ is small enough and $t$ is small enough
  (each depending only on the bounds of the metric and the previous
  constants), then  for the associated
        $t$-norm and for families of kernels $K_1,$ and $K_2$,   
	\be
	\norm{K_i}_\mathrm{op} \leq e^{2B\sqrt{t}} \norm{K_i}_{\tnm{t}},
	\ee{top-bd}
	\be
	\abs{K_i(x,y;t)} \leq  \norm{K_i}_\tnm{t}\sbrace{2 (2 \pi t)^{-m/2} e^{-d(x,y)^2/(4t)} + t e^{-D^2/(20t)}};
	\ee{tpntws-bd}
 in particular, there is an $A_2>0$ such that 
	\be
	\norm{K_i(t)}_\infty \leq  A_2t^{-m/2}\norm{K_i}_\tnm{t}.
	\ee{tinf-bd}
        Finally, 
	\be \norm{K_1(t_1)*K_2(t_2)}_\tnm{t}\leq
        e^{B\sqrt{t}}\norm{K_1}_\tnm{t_1}\norm{K_2}_\tnm{t_2}.
        \ee{t*-bd}
\end{corollary}
\begin{proof}
  Eqs.~\eqref{eq:top-bd},
\eqref{eq:tpntws-bd}  and \eqref{eq:t*-bd} of the corollary are
simply restatements of Eqs.~\eqref{eq:kop-bd}, \eqref{eq:inf-bd} and~\eqref{eq:*-bd} of
the proposition. Eq.~\eqref{eq:tinf-bd} is a separately-useful
immediate consequence of Eq.~\eqref{eq:tpntws-bd}.
\end{proof}

\subsection{Approximate semigroups and approximate kernels}
As noted in the introduction, the definition of approximate semigroup
below will ensure that the fine-partition limit of kernel
products of approximate semigroups converge.  The definition of
approximate heat kernel will ensure that it is an
approximate semigroup and that the fine-partition limit of its kernel
products is in fact the
heat kernel of the associated operator. 
\begin{definition} \label{def:approx-semigroup}
A family of kernels $K(t)$ is an \emph{approximate semigroup} with
constants $(B,C,D,T)$ if  for $0<t_1,t_2$ and $t=t_1+t_2<T$ 
with the $t$-norm of Def.~\ref{def:t-norm} 

\be
	 \norm{K(t)}_{(t)}\leq 1
\ee{k-bd}
and
\be
	\norm{K(t_1)*K(t_2)-K(t)}_{(t)}\leq Ct^{3/2}.
\ee{k*k-est}
\end{definition}

\begin{remark}
Note that Eq.~\eqref{eq:k-bd} implies an approximate semigroup $K(t)$ must be
in $\EE'(t)$ for all $t<T$. Moreover, accordingly writing $K(t)=\widetilde{K}(t) + J(t)$ for $\widetilde{K}(t) \in \EE(t),$
the following lemma says 
it suffices to check Eq.~\eqref{eq:k*k-est} only on
$\widetilde{K}(t)$. 
\end{remark}
\begin{lemma} \label{lm:k*k-est}
If $K(t) = \widetilde{K}(t) + J(t)\in \EE'(t)$ with  $\widetilde{K}(t) \in \EE(t)$  satisfying
Eq.~\eqref{eq:k*k-est}, then $K(t)$ satisfies Eq.~\eqref{eq:k*k-est},
albeit with potentially smaller $D,$ larger $C,$ and smaller $T$. 
\end{lemma}
\begin{proof}
  Consider
  \begin{align}
    & (\widetilde{K}  + J)(t_1)*(\widetilde{K}  + J)(t_2)  - (\widetilde{K}  + J)(t) =
    \nonumber \\
    & \qquad \qquad \widetilde{K} (t_1)*\widetilde{K} (t_2) - \widetilde{K} (t) +
    \widetilde{K} (t_1)*J(t_2)\nonumber  \\
    & \qquad \qquad + J(t_1)*\widetilde{K} (t_2) +  J(t_1)*J(t_2) - J(t) .\nonumber
  \end{align}
  By hypothesis, the  first two terms on the right-hand side combine to give $Ct^{3/2}$
  times an element of $\EE'(t)$. Applying Eq~\eqref{eq:kopnorm-bd}
  bounds  each of the next two terms by 
  \[
  e^{1.1 B\sqrt{t}} t
  e^{-D^2/(20t)}.
  \] 
  Replacing $D$ with $ D/2$,
  these terms are thus each
  bounded by $t^{5/2}  e^{-D^2/(20t)}$, for small $T$. Easy
  estimates give the same bound for the remaining two
  terms, so the sum on the right-hand side, after division by $(C + 3) t^{3/2}$ lies  in $\EE'$.
\end{proof} 

\begin{definition} \label{def:ahk}
  Let $\Delta$ denote a second order elliptic differential operator
  defined on $O \in \RR^m$ acting on functions with  values in
  $\RR^n.$
  Suppose the second-order coefficients of $\Delta$ are  the inverse of
  the metric $g$ 
  (i.e., $\Delta$ is a generalized Laplacian)
  and the lower-order   coefficients  are bounded in sup
  norm. A family of  kernels $K(t)$ is an \emph{approximate heat kernel}
  for $\Delta$ with constants $(B,C,D,T)$, all positive, if  it is
  differentiable to first order in $t \in (0,T)$ and to second order in the
  spatial variables, and if, for $t<T$ and using the $t$-norm with constants $(B,D),$
  \be
  \norm{K(t)}_{(t)}\leq 1,
  \ee{k2-bd}
  for all $f\colon O \to \RR^n$
  \be
  \lim_{t \to 0} K(t)*f =f, 	
  \ee{tto0}
  \be 
  \lim_{t \to 0} \frac{K(t)*f -f}{t}  = \frac{\Delta}{2} f 
  \ee{kf}
  (both pointwise),
  \be
  \norm{\frac{\partial}{\partial x}K(x,y;t)}_{(t)},\norm{\frac{\partial}{\partial y}K(x,y;t)}_{(t)}\leq B/t,
  \ee{kprime-bd}
  and
  \eqa{ahk}{
    \norm{\parens{\frac{1}{2}\Delta_x  - \frac{\partial}{\partial t}}
      K(x,y;t)}_\tnm{t} & \leq  Ct^{1/2}\nonumber \\
    \norm{\parens{\frac{1}{2}\Delta^*_y  - \frac{\partial}{\partial
          t}} K(x,y;t)}_\tnm{t} & \leq Ct^{1/2},
  }
  where
  $\Delta_x$ acts from the left on  $\End(\RR^n)$ and  $\Delta_y^*$ acts
  from
  the right via $\int_{O} \Delta_y^*[h^*(y)] \cdot f(y)
  \dmeas_g y = \int_{O} h^*(y) \cdot \Delta_y[f(y)]
  \dmeas_gy.$
 \end{definition}


\begin{proposition}\label{pr:approx-heat}
  Suppose $K(t)$ is an approximate heat kernel for the elliptic
  operator $\Delta$ and metric $g$ with constants $(B,C,D,T).$  
  Then there exist positive constants $B_1,C_1,D_1,T_1$ (each depending
  on the bounds of the metric and $\Delta$, on $B,C,D,T$ and on the previous
  constants) such that $K$ is an approximate semigroup with constants
  $(B_1,C_1,D_1,T_1)$. 
\end{proposition}

 \begin{proof} 
   Make $B$ large enough, and $D$ and $T$ small enough that
   Prop.~\ref{pr:t-def}, Cor.~\ref{cr:t-norm} and Lemma~\ref{lm:k*k-est} hold.
   According to Lemma~\ref{lm:k*k-est}, it suffices to prove
   Eq.~\eqref{eq:k*k-est} for $K(x,y;t)
        \in \EE_{B,D}(t).$   For $d(x,z) \geq D/2$ 
        Eqs.~\eqref{eq:t*-bd} and~\eqref{eq:tpntws-bd} imply
	\[\abs{\sbrace{K(t_1) *  K(t_2)}(x,z)} \leq 2e^{Bt^{1/2}}(2 \pi t)^{-m/2}e^{-D^2/(16t)}+ t^2e^{-D^2/(20t)} \leq C_1t^{5/2}e^{-D_1^2/(20t)}\]
	for large enough $C_1,$ and small enough  $D_1$ and $t$,
        giving Eq.~\eqref{eq:k*k-est}. 

	For $d(x,z)
	\leq  D/2,$ the left hand side of Eq.~\eqref{eq:k*k-est} is
	  \begin{align*}
	    & \norm{K(t_1) *  K(t_2)- K(t)}_{\tnm{t}}  \leq   \int_{0}^{t_1} \norm{\frac{\partial}{\partial
	      \tau}\sbrace{K (\tau) *   K (t-\tau)}}_{\tnm{t}} \dmeas \tau\\
	    &= \int_{0}^{t_1}   \norm{\dot{K}(\tau) *   K (t-\tau)- K (\tau)  *   \dot{K}(t-\tau)}_{\tnm{t}} \dmeas \tau\\
	    & \leq \int_{0}^{t_1}
	    \norm{\frac{1}{2}\cbrace{\Delta_{y}^* \sbrace{K(\tau)}
	       *   K(t-\tau)  -  K(\tau) *   \Delta_{y}\sbrace{K(t-\tau)} }}_{\tnm{t}} \dmeas \tau \\
	    &\qquad +\int_{0}^{t_1}
	    C\tau^{1/2} e^{Bt^{1/2}}  +
	      C(t-\tau)^{1/2} e^{Bt^{1/2}} \dmeas \tau \\
            &  \leq 
	    \norm{\int_{0}^{t_1}\frac{1}{2}\int_{\partial_y}
	    \delta_y \sbrace{K(\tau)} \cdot_y K(t-\tau) \,\dmeas_g y
            \, \dmeas \tau }_{\tnm{t}} +
	    \frac{4}{3} C e^{B t^{1/2}} t^{3/2}
	  \end{align*}
          where  the third  inequality uses
          Eqs.~\eqref{eq:ahk} and~\eqref{eq:k2-bd} of the definition of
	  an approximate heat kernel and Eq.~\eqref{eq:t*-bd}
	  of Prop.~\ref{pr:t-def}, and  the subscript $y$ indicates
	  the operators act on the fiber over the middle copy of
	  $\RR^n$ (the one that $*$ contracts over). 
          The first term of the last
	  equation, call it $\norm{J(x,z;t)}_{(t)},$ is the 
	  boundary term obtained using the formal adjoint of
	  $\Delta.$ That is,  $\delta$ is the first order
          operator for which
	  \[
          \int_{R} \parens{f \cdot \Delta h - \Delta^*f \cdot h}
          \dmeas_g y = \int_{\partial R} \delta f \cdot h \, \dmeas_g y ,
          \]
	   and the subscript means the integral is over $y$ such
           that one of
	 $d(x,y)$ and $d(y,z)$ is equal to $D$ and the other less.  So, since
	  $d(x,z) \leq D/2,$  both $d(x,y) \geq D/2$ and
	  $d(y,z) \geq D/2.$ 
          The boundary integral (notice it has finite
          volume with a bound depending on $D$ and bounds on the first
          two derivatives of the metric)  can thus be bounded by a
          multiple of
          $e^{-D^2/(20t_1)}e^{-D^2/(20t_2)}/\sbrace{P_1(t_1)P_2(t_1)}$
          where $P_i$ are polynomials (uses Eq.~\eqref{eq:kprime-bd}).
          Therefore $\norm{J(t)}_\infty \leq ct^{5/2}e^{-d^2/t}$ for some
          $c,d.$ 

For fixed $x$ the set of $z$ for which $J(x,z;t)$ is
          nonzero is a ball of radius $D/2$ (and likewise for $x$ and
          $z$ reversed) which has bounded volume (depending on $D$ and
          the bounds of the first two derivatives of the metric), so
          also $\norm{J(t)}_{\mathrm{op}}= c't^{5/2}e^{-d^2/t}$ and
          therefore $\norm{J(t)}_{\tnm{t}} \leq C_1t^{3/2}$ with the
          appropriate constants in the definition of the $t$-norm.
          Eq.~\eqref{eq:k*k-est} 
	  follows. 

\end{proof}
 \subsection{Manifolds}
 
 \begin{definition}  \label{def:tame}
   Suppose $\VV$ is an $n$-dimensional vector bundle over an
   $m$-dimensional
manifold $M$ with Riemannian metric $g.$
   An atlas of charts
   for $\VV$ over $M$ is 
   \emph{tame} if  
   \begin{itemize}
   \item All derivatives of $g$ and $g^{-1}$ expressed in coordinates of order $0 \leq k \leq 6$  are uniformly bounded in sup norm on all charts.
   \item There is a $D_0>0$ such that the ball of radius $D_0$ around any point is contained in a single chart.
   \end{itemize}
   The tuple $(M,g,\VV)$ is \emph{tame} if it admits a tame atlas.  
   If  $\Delta$ is a 
   generalized Laplacian, i.e. 
   a second-order elliptic
   operator on sections of $\VV$ 
   which in local coordinates is
 of the form
   \[\Delta = g^{ij} \frac{\partial^2}{\partial x_i \partial x_j} + A^i
   \frac{\partial}{\partial x_i} + B\]
   (with $A^i$ and $B$ 
   valued in   $\operatorname{Matrix}_{n,n}$),
   and if there is a tame atlas so that the 
   derivatives of order $0 \leq k \leq 2$ of  $A^j$ and $B$ 
   in all charts  are uniformly bounded in sup norm, then
   $(M,g,\VV,\Delta)$ is 
   \emph{tame.} 
\end{definition}

Of course any such data is tame if  $M$ is compact and everything is
smooth.

Let $\pi_i : M \times M \to M$ be the projection onto the $i$th copy
of $M$, and consider the bundle $\Hom_{yx}$ of homomorphisms from
$\pi_2^* \VV$ to $\pi_1^* \VV$. Its fiber over $(x,y)$ is
$\Hom(\VV_y,\VV_x)$. Call a section $K(x,y)$ of $\Hom_{yx}$
 a \emph{kernel on $\VV$.} 
$K$ is then a kernel in the
sense of the Subsection~\ref{ss:coords} on any chart for
$\VV$ (where $m$ is the dimension of $M$ and $n$ the dimension of
$\VV$). On any tame atlas, for sufficiently large $B$ and sufficiently
small $D$, there is a sufficiently small $t$ such that
the $t$-norm with constants $(B,D)$ can be defined on each chart,
and thus it makes sense to define $\norm{K}_{\tnm{t}}$ to be the
supremum of the $t$-norms of its image in each chart.

\begin{corollary} 
  If $(M,g,\VV)$ is tame the $t$-norm defined in terms of any tame
  atlas will satisfy Eqs.~\eqref{eq:top-bd}-\eqref{eq:t*-bd}
  for
  sufficiently large $B$ and sufficiently small $D$.
\end{corollary}

 \begin{definition} \label{def:approx-ker} 
A family of kernels $K(t)$ on $\VV$ for $t>0$ is an
\emph{approximate semigroup} with constants $(B,C,D,T)$ if $(M,g,\VV)$
admits a tame atlas on each chart of which $K$ is  represented as an approximate semigroup with constants $(B,C,D,T),$
with $D\leq D_0$ above.  A family of kernels $K(t)$ on  $\VV$ is
an \emph{approximate heat kernel} with constants $(B,C,D,T)$ if $(M,g,\VV)$
admits a tame atlas on each chart of which $K$ is represented as an approximate heat kernel with
constants $(B,C,D,T)$ with  $D \leq D_0.$  
  
\end{definition} 
\begin{corollary} \label{cr:manifolds}
An approximate semigroup on a vector bundle satisfies
Eqs.~\eqref{eq:t*-bd}-\eqref{eq:k*k-est}.  An approximate heat kernel 
for some $\Delta$ on $\VV$ is an approximate semigroup, with 
constants $(B,C,D,T)$ for the approximate semigroup whose constants
 can be made to depend
only on the corresponding constants for the approximate heat kernel
and the bounds on the defining atlas.
\end{corollary}

\begin{remark}
  While it suffices for the rest of the work, the dependence of the
  structures defined on the choice of tame atlas is mathematically
  distressing.  In fact there is a natural notion of the comparability
  of tame structures, which simply involves requiring that the
  diffeomorphisms between charts induced by the identity on $\VV$ have
  all derivatives up to the appropriate order uniformly bounded.  It
  is then straightforward if laborious to check that the $t$-norms
  associated to compatible tame atlases are comparable (each bounded
  by a multiple of the other), that families of kernels that are
  approximate semigroups or heat kernels with respect to one atlas
  are the same with respect to the other, and therefore that the limit
  results of the following section depend only on the ``tame
  equivalence class'' of the vector bundle, Riemannian manifold and
  operator.
\end{remark}

 \section{The fine-partition limit} \label{sc:limit}

 If $P=(t_1, t_2,\ldots, t_k)$
  is a
 \emph{partition} of a positive real number
 $t$ (that is, $t_i>0$ and $\sum_i t_i=t$) define $\abs{P}=\max_i
 t_i,$ $\#P=k,$ and for any kernel $K$ 
 \be
 K^{*P}(t)=K(t_1) *  K(t_2) *  \cdots  *  K(t_k).
 \ee{k*p-def}

  If $P$ is a  partition of $t$ and
 $P'$ is a partition  of $t',$ then the concatenation $PP'$ is a
 partition of $t+t';$  if $P_i$ is a partition of $t_i$ for $1 \leq
 i \leq k,$ then the partition $P_1P_2 \cdots P_k$ is a \emph{refinement} of $P=(t_1,
 \ldots, t_k).$

   In the language of the introduction, 
   $K^{*P}$ is the approximate path
   integral corresponding to the approximate heat kernel $K$ and a
   choice of partition $P$. Thm.~\ref{th:kinf} 
   below asserts the
   convergence of these approximations and provides a key estimate on
   the rate of convergence in terms of $t$ and $\abs{P}$, valid
   provided $K$ is an approximate semigroup in the precise sense of
    Defs.~\ref{def:approx-ker}
   and~\ref{def:approx-semigroup}.

              \subsection{Partitions and the refinement limit}\label{ss:partition}

 \begin{lemma} \label{lm:part-est}
   Suppose  $K(t)$ is a family of kernels and
   $\norm{\,\cdot\,}_{\tnm{t}}$ is a family of norms for which
   Eqs.~\eqref{eq:t*-bd},~\eqref{eq:k-bd}, and~\eqref{eq:k*k-est} hold
   for some constants $B,C,$ and $T$. Then there is an $A>0$ depending
   on $B,C$ such that, if  $T$ is chosen small enough, 
   \be 
   \norm{K^{* Q}(t)-K^{* P}(t)}_{\tnm{t}}< At^{5/4} \abs{P}^{1/4} 
   \ee{refinement}
 for all  refinements $Q$ of all
 partitions $P$ of $t<T.$
 \end{lemma}

 \begin{proof}
   First observe that by Eqs.~\eqref{eq:t*-bd},~\eqref{eq:k-bd}, and~\eqref{eq:k*k-est} 
   there is a  $c_2>0$ so that for all sufficiently small $t=t_1+t_2+t_3$
   \[\norm{K(t_1) *  K(t_2) *  K(t_3)-K(t)}_{\tnm{t}}\leq
   c_2t^{3/2}.\]

    Next, argue by induction on the
   number of entries in $Q$ that there are positive reals $b_2,c_3>0$ such that 
   \be \norm{  K^{* Q}(t)-K(t)}_{\tnm{t}}\leq
   c_3e^{b_2t^{1/2}}t^{3/2}. 
    \ee{kq-k}
  For that note one can always write
   $Q=Q_1(t_2)Q_3,$ where $Q_1$ is a partition of $t_1$ and $Q_3$ is a partition
   of $t_3,$ $t_2$ is a component of $Q$,  $t_1\leq t/2$ and $t_3\leq t/2$ (one or both of $t_1,t_3$ may be $0$).  Then
   \begin{align*}
   &  \norm{K^{* Q}(t)-K(t)}_{\tnm{t}} \\
 &\leq
     \norm{\sbrace{K^{* Q_1}(t_1)-K(t_1)} *  K(t_2) *
       K(t_3)}_{\tnm{t}}\\
& \qquad   +\norm{K(t_1) *  K(t_2) *  \sbrace{K^{* Q_3}(t_3)-K(t_3)}}_{\tnm{t}}\\
     &\qquad + \norm{\sbrace{K^{* Q_1}(t_1)-K(t_1)} *  K(t_2) *  \sbrace{K^{* Q_3}(t_3)-K(t_3)}}_{\tnm{t}} + \norm{K(t_1) *  K(t_2) *  K(t_3)-K(t)}_{\tnm{t}}\\
     &\leq e^{2Bt^{1/2}}\norm{K^{* Q_1}(t_1) -K(t_1)}_{\tnm{t_1}} +
       e^{2B t^{1/2}}\norm{K^{* Q_3}(t_3)-K(t_3)}_{\tnm{t_3}} \\
 & \qquad + e^{2Bt^{1/2}}\norm{K^{* Q_1}(t_1) -K(t_1)}_{\tnm{t_1}} \norm{K^{* Q_3}(t_3)-K(t_3)}_{\tnm{t_3}}
     + c_2t^{3/2}\\
     &\leq e^{2Bt^{1/2}}\sbrace{c_3e^{b_2 t_1^{1/2}} t_1^{3/2}+ c_3
       e^{b_2 t_3^{1/2}} t_3^{3/2}
       + c_3^2e^{b_2(t_1^{1/2}+t_3^{1/2})} t_1^{3/2}t_3^{3/2}+
       c_2t^{3/2}}\\
     &\leq c_3t^{1/2}e^{2B t^{1/2}+
       b_2\parens{t/2}^{1/2}}\sbrace{t2^{-1/2}  + c_3e^{b_2 t^{1/2}}t^{5/2}+ \parens{c_2/c_3}t}
       \\
     &\leq c_3 t^{1/2}e^{b_2 t^{1/2}}\parens{2^{-1/4}t + c_3e^{b_2
         t^{1/2}}t^{5/2}} \leq c_3 e^{b_2 t^{1/2}}t^{3/2}
    \end{align*}
	where the second inequality follows from Eqs.~\eqref{eq:t*-bd}
        and~\eqref{eq:k-bd}, the third from the inductive hypothesis,
        the fourth from $t_1<t/2,$ $t_3<t/2,$  the fifth from choosing $b_2>2B/(1-2^{-1/2})$ and $c_3>c_2/(2^{-1/4}-2^{-1/2})$
        (this condition also covers the base case) and the last by
        setting $T$ small enough that  $c_3e^{b_2 t^{1/2}}t^{3/2}$ is
        less than $(1-2^{-1/4}).$  

Note  this implies for $K$ as above
there is a $b_3>0$ so that for any kernel $J$, for $t=t_1+t_2<T$ for
small enough $T,$ and any partition $Q$ of $t_1$
 \eqa{kpop-tbd}{
 \norm{K^{* Q}(t_1) *  J(t_2)}_{\tnm{t}}&\leq
 e^{b_3 t^{1/2}}\norm{J}_{\tnm{t_2}} \mbox{, and} \nonumber\\
\norm{K^{*Q}(t)}_{\tnm{t}}&\leq e^{b_3 t^{1/2}}.
}
This follows from simply writing $K^{* Q}$ as $K^{* Q} - K + K$.

    Now  argue by induction 
  on the number of entries in $P$ 
  that there is an $A>0$ such that
 for all partitions $P$ and all refinements $Q$ of $P$ 
  \[\norm{K^{* Q}(t)-K^{* P}(t)}_{\tnm{t}} < At^{5/4}\abs{P}^{1/4}.\]
  Suppose
  $t=t_1+t_2+t_3,$   $P=P_1(t_2)P_3$ with $P_i$ a partition of
 $t_i,$ 
   and $Q=Q_1Q_2Q_3$ with $Q_i$ a refinement of $P_i,$  $Q_2$ a
   partition of $t_2,$  chosen so that $t_1<t/2$ and $t_3<t/2$ (one or both of $t_1,t_3$ may be $0$).
   Then
   \begin{align*}
     &\norm{K^{* Q}(t)-K^{* P}(t)}_{\tnm{t}} =
     \norm{K^{* Q_1}(t_1) *  K^{* Q_2}(t_2) *  K^{* Q_3}(t_3)-K^{* P_1}(t_1) *  K(t_2) *  K^{* P_3}(t_3)}_{\tnm{t}}\\
       &\qquad \leq 
       \norm{\sbrace{K^{* Q_1}(t_1)-K^{* P_1}(t_1)} *  K^{* Q_2Q_3}(t_2+t_3)}_{\tnm{t}} 
       +  \norm{K^{* P_1}(t_1) *  \sbrace{K^{* Q_2}(t_2)-K(t_2)} *  K^{* Q_3}(t_3) }_{\tnm{t}} \\
 & \qquad \qquad 
        +  \norm{K^{* P_1}(t_1) *  K(t_2) *  \sbrace{K^{* Q_3}(t_3) - K^{* P_3}(t_3)} }_{\tnm{t}} \\
       &\qquad \leq     e^{2b_3  t^{1/2}}\sbrace{ A t_1^{5/4}
         \abs{P}^{1/4} + c_3 e^{b_2 t^{1/2}}  t_2^{3/2}+ At_3^{5/4}\abs{P}^{1/4} }\\
       &\qquad \leq    Ae^{2b_3 t^{1/2}} t^{5/4} \abs{P}^{1/4}
       \sbrace{ 2^{-1/4} + c_3A^{-1}e^{b_2 t^{1/2}}}\\
         & \leq
       Ae^{2b_3 t^{1/2} } t^{5/4} \abs{P}^{1/4}2^{-1/8} \leq At^{5/4} \abs{P}^{1/4} 
       \end{align*}
      where the second inequality follows from Eqs.~\ref{eq:kq-k}
      and~\ref{eq:kpop-tbd} with the
      inductive hypothesis, and the
      third from $t_1<t/2,$ $t_3<t/2$ and $t_2^{3/2} \leq
      t^{5/4}\abs{P}^{1/4}.$  The fourth inequality follows by
      choosing $A$ large enough and $T$ small enough that $A> c_3
      e^{b_2 t^{1/2} }/(2^{-1/8}-2^{-1/4})$ (a similar choice covers the base
      case), and the last follows by choosing $T$ small enough that
      $e^{2b_3 t^{1/2} }<2^{1/8}.$ 
 \end{proof}

 \begin{theorem} \label{th:kinf}
  Suppose  $K(t)$ is an approximate semigroup with constants $(B,C,D,T)$ on a bundle $\VV$ 
  over $M$ with   Riemannian
 $g,$ all tame, and $\norm{\,\cdot\,}_{\tnm{t}}$ is the norm on
 kernels on $\VV$ guaranteed by Prop.~\ref{pr:t-def}. 
Then there is a
 family of kernels $K^\infty (t)$ and a constant 
  $A>0$ depending on $B,C$   such that if  $T$ is chosen small enough 
  \be
  \norm{K^\infty(t) - K^{* P}(t)}_{\tnm{t}} \leq  A t^{5/4}\abs{P}^{1/4}
  \ee{kinf-test}
  for any partition $P$ of any $t<T.$ In particular 
  $K^\infty$ can be extended to all $t>0$ and there are $A_1,T_1,$ $D_1$ depending only on these constants (and dimensions) such that
 \be
 \norm{K^\infty(t) - K^{* P}(t)}_{\infty} \leq  A_1te^{B_1t} \abs{P}^{D_1}
 \ee{kinf-est}
 for all $P$ with $\abs{P}<T_1,$ so that in fact 
  \be
  K^\infty(t) = \lim_{\abs{P} \to 0} K^{*P}(t)
  \ee{kinf-lim}
 in supremum norm for each fixed $t>0.$
 \end{theorem}

 \begin{proof} 
	 For the short-time construction of $K^\infty$ and
         Eq.~\eqref{eq:kinf-test}, consider a sequence  $P_1 = (t),
         P_2, \ldots$ 
         for sufficiently small $t$,
         each a refinement of the previous and  with $\abs{P_i} \to 0$.
         By Eq.~\eqref{eq:refinement}, $K^{*P_i}(x,y,t)$  is a Cauchy
         sequence in the $t$-norm, and therefore by
         Eq.~\eqref{eq:tinf-bd} is Cauchy in supremum norm and by
         completeness converges to some $K^\infty(x,y;t)$. 
         If $P$ is any 
         partition of $t$, let $P_i'$ be a  refinement of $P_i$ and   $P$  for
         each $i>1$.
         Then
         \begin{align*}
           \norm{K^{*P}(t)-K^\infty(t)}_{\tnm{t}} & \leq
           \norm{K^{*P}(t)-K^{*P'_i}(t)}_{\tnm{t}}  +
           \norm{K^{*P'_i}(t)-K^{*P_i}(t)}_{\tnm{t}}  +
           \norm{K^{*P_i(t)}-K^\infty(t)}_{\tnm{t}}\\
           &< At^{5/4}\abs{P}^{1/4} + 2At^{5/4}\abs{P_i}^{1/4} \\
           &\leq At^{5/4}\abs{P}^{1/4},
         \end{align*}
         taking $i$ to infinity.
         This proves Eq.~\eqref{eq:kinf-test}.

         Eq.~\eqref{eq:kinf-est} and hence Eq.~\eqref{eq:kinf-lim} will follow for an approximate
         semigroup $K$  from the observation that, for
         $P$ a sufficiently fine partition of a given arbitrary  $t>0$ and $Q$ any
         refinement of $P,$ there are  constants $A_1,b_1>0$ such that 
         \be
         \norm{K^{* Q}(t) - K^{* P}(t)}_{\infty} \leq
         A_1te^{b_1 t} \abs{P}^{1/(8m)}. 
         \ee{kq-est}

         To see Eq.~\eqref{eq:kq-est} suffices, consider a
         sequence $P_i$ of partitions with $\abs{P_i} \to
         0$.  Consider any two $P_{i_1}, P_{i_2}$ 
         with $i_1<i_2$ far enough out in the sequence for
         Eq.~\eqref{eq:kq-est} to apply, and let $Q$ be  a common 
         refinement. Then  the bounds on $\norm{K^{* Q}(t) - K^{*
             P_i}(t)}_{\infty}$ imply $\norm{K^{* P_{i_1}}(t) - K^{*
             P_{i_2}}(t)}_{\infty} \in \OO\parens{\abs{P_{i_1}}^{1/(8m)}}$ for
         fixed $t$.  Thus the sequence is Cauchy in the supremum norm, so a limit
         $K^\infty(t)$ exists. For $P$ in the given sequence, the obvious
         estimate shows the limit satisfies Eq.~\eqref{eq:kinf-est}, which is
         the crux of the theorem. If $t$ is small, this limit clearly agrees with  the
         short-time construction above, and the argument above extends to
         show Eq.~\eqref{eq:kinf-est} 
         and therefore Eq.~\eqref{eq:kinf-lim} 
         in fact follow from
         Eq.~\eqref{eq:kq-est} for all partitions.

         To see  Eq.~\eqref{eq:kq-est} holds, let $P$ be a partition of $t >
         0$, where $t$ need not be particularly small, and let $T$ be small
         enough that Cor.~\ref{cr:t-norm}
         holds. Assume
         $\abs{P}^{1/(2m)}<T$ and $\abs{P} <1$, so
         $\abs{P} < \abs{P}^{1/(2m)}$. Let $P_0$ be another partition of $t$
         such that $P$
         is a refinement of $P_0$ and such that each component $t_j$ of $P_0$ satisfies
         $\abs{P}^{1/(2m)} \leq t_j \leq 2 \abs{P}^{1/(2m)}.$ (To define
         $P_0$,
         proceed inductively, using
         $\abs{P}<\abs{P}^{1/(2m)}$) If the partition $P$ is sufficiently
         fine, then the upper bound on $t_j$ will ensure
         Eqs.~\eqref{eq:refinement} and~\eqref{eq:kop-bd} hold with $t_j$
         replacing the generic $t$ in these equations. For each $t_j$ in  $P_0,$ the partitions $P$ and
         $Q$ restrict to partitions $P_j$ and $Q_j$ respectively of $t_j$. In terms of these,
         \eqa{*diff}{
           K^{* Q}(t)-K^{* P}(t) =  
           & \sum_j K^{* P_{1}}(t_1) * 
           \cdots  *  K^{* P_{j-1}}(t_{j-1}) * 
           \sbrace{K^{*Q_{j}}(t_j)-K^{* P_{j}}(t_j)} \nonumber\\
           &\qquad * K^{* Q_{j+1}}(t_{j+1})
           *  \cdots  *  K^{* Q_{k}}(t_{k}).
         }
         Eq.~\eqref{eq:refinement}
         together with Eq.~\eqref{eq:tinf-bd} and the  bounds on $t_j$ give
         \begin{align*}
           \norm{K^{* Q_{j}}(t_j) - K^{* P_{j}}(t_j)}_{\infty}
           &\leq A_2 t_j^{-m/2}\norm{K^{* Q_{j}}(t_j) -
             K^{* P_{j}}(t_j)}_{\tnm{t_j}} \\
           &\leq  A_3 t_j^{5/4} \leq
           A_4 t_j \abs{P}^{1/(8m)},
         \end{align*}
         where  the second inequality follows from $\abs{P_j} \leq \abs{P}$ and
         the lower bound on $t_j$, while the final inequality follows from the
         upper bound on $t_j.$

$P_0$ is a refinement of a partition $(\tau_1,t_j, \tau_2)$, and
         restricts to partitions $P_{0,1}$ and $P_{0,2}$ of $\tau_1$
         and $\tau_2$.  If $\tau_1<T$ for small enough $T$  then
         $\norm{K^{*P_{0,1}}(\tau_1)}_{\text{op}}\leq
         e^{(2B+b_3)T^{1/2}}\leq c$ for some $c > 1$, by
         Eqs.~\eqref{eq:top-bd} and~\eqref{eq:kpop-tbd}. In this case,
         $\norm{ K^{* P_{1}}(t_1) * \cdots * K^{*
             P_{j-1}}(t_{j-1})}_{\text{op}} =
         \norm{K^{*P_{0,1}}(\tau_1)}_{\text{op}} \leq c$. 
         On
         the other hand, if $\tau_1 \geq T$, 
         since $t_j<2T$ there is a
         partition $P'_{0,1} = (t'_1, \cdots, t'_p)$ of $\tau_1$ which
         is a refinement of $P_{0,1}$ with each $t_i'$ satisfying $2T<
         t_i' <4T$ (keep combining adjacent $t_i$ until it is between
         these two limits). The bounds on $t'_j$ ensure
         $\norm{K^{*P'_{0,1}(t_j')}}_{\text{op}} \leq
         e^{(2B+b_3)(t_j')^{1/2}} \leq e^{b_1 t_j'}$  for $b_1 = 2^{-1/2} (2B +
           b_3) T^{-1/2}$, by Eqs.~\eqref{eq:top-bd}
           and~\eqref{eq:kpop-tbd}. This ensures $\norm{ K^{*P_{1}}(t_1) * \cdots * K^{* P_{j-1}}(t_{j-1})}_{\text{op}}
           \leq e^{b_1 \tau_1}$. The cases $\tau_2 < T$ and   $\tau_2
           \geq T$ lead to an
           analogous estimate. 
         
         Combining the above estimates, Eq.~\eqref{eq:*diff} gives 
         \begin{align*}
           &\norm{K^{* Q}(t)-K^{* P}(t)}_\infty \\
           &\qquad \leq
           \sum_j  \norm{K^{* P_{1}}(t_1) *  \cdots  *  K^{* P_{j-1}}(t_{j-1})*
             \sbrace{K^{*
                 Q_{j}}(t_j)-K^{* P_{j}}(t_j)} \right. \right.\\
         &\qquad \qquad \left.\left. * K^{* Q_{j+1}}(t_{j+1}) *    \cdots  *  K^{*
               Q_{k}}(t_k)}_\infty \\ 
           &\qquad \leq  \sum_j c^2 e^{b_1 t} \norm{K^{*
               Q_{j}}(t_j)-K^{* P_{j}}(t_j) }_\infty\\ 
           &\qquad \leq  c^2 e^{b_1 t}  \sum_j A_4 t_j \abs{P}^{1/(8m)}  \leq
           A_1 t e^{b_1 t}\abs{P}^{1/(8m)}.
         \end{align*}
  
 \end{proof}

              \subsection{Relating different kernels}\label{ss:limit}
Sect.~\ref{sc:kernel}
defined approximate semigroups locally.  Even
local values of the fine partition limit depend globally on the value
of the approximate semigroup, but the next proposition shows this to
be true rather weakly. In fact, changing the kernel or even the
underlying manifold outside a neighborhood of a point, as long as the bounds
$B,C$ and $D$ remain fixed, will change the fine partition limit at the given
point only by an exponentially damped term.  

\begin{proposition} \label{pr:local}
Suppose  $(\VV_i,M_i,g_i,K_i(t))$ for $i=1,2$ each represent a tame
bundle over a tame Riemannian manifold with an approximate semigroup.
Suppose $x$ is a point in $M_1,$ $O$ is a neighborhood of that point,
and $\Phi$ is an isomorphism of all this structure to an open set
$\Phi[O] \subset M_2.$  That is to say in a neighborhood of $x$ the
bundle $\VV_1$ can be identified isometrically with the bundle
$\VV_2$ over a neighborhood of $\Phi(x)$ so that $K_1(t)$ is the
pullback of $K_2(t).$   Then there are constants $c,d>0$ depending
only on the distance $r$ of $x$ to the boundary of  $O$ and on the constants $(B,C,D)$ associated to the two
approximate semigroups such that for some $T$ depending on these
$(B,C,D)$ and all $0<t\leq T,$ 
\be
\abs{K_1^\infty(x,x;t) -K_2^\infty(\Phi(x),\Phi(x);t)} \leq ce^{-d/t},
\ee{local}
which is to say their difference is exponentially damped.  
\end{proposition}

\begin{proof}  Choose constants $(B,C,D)$ and the associated $T$ for
  which both kernels are
  approximate semigroups. We can assume $r<D$ and   $O$ (resp. $\Phi[O]$) contains a
  ball of radius $r$ around $x$ (resp. $\Phi(x)$) in $g_1$
  (resp. $g_2$).  For simplicity of notation, write $K$ for $\Phi^* K$
  and $O$ for $\Phi[O].$     Since Eq.~\eqref{eq:local} is obvious for
  large $t$ by Eq.~\eqref{eq:tinf-bd}, assume $t<T.$ 
  Let $\chi(y)$ be a real valued function on both $M_i$ which is $1$
  if $y \in O$ and $0$ otherwise.  For simplicity write $\chi K$ for
  the kernel whose value on $y,z$ is $\chi(y)K(y,z)$ and $K \chi$ for
  the one whose value at $y,z$ is $K(y,z)\chi(z).$  All this
  notational slight of hand and the fact that $K_1=K_2$ on $O \times
  O$   allows both sides of the following
  expression to make unambiguous sense and to be equal: 
  $(1-\chi) K_1 \chi  -  \chi  K_2 (1-\chi) =  K_1 \chi  - \chi K_2$.
  This implies, for
  any partition $P$ of $t<T,$ the
  sum $\sum_j K_1^{*P_j} * (1-\chi) K_1(t_j) \chi * K_2^{*P_{j'}} -
  K_1^{*P_j} * \chi K_2(t_j)(1- \chi) * K_2^{*P_{j'}},$ where $P = P_j(t_j)P_{j'},$ telescopes to
  $K_1^{*P} \chi - \chi K_2^{*P},$ 
  which reduces to $K_1^{*P} -
  K_2^{*P}$ on $O \times O.$ Turning this around,
  \[ \abs{K_1^{*P} - K_2^{*P}} \leq \sum_j \abs{K_1^{*P_j}* (1-\chi) K_1(t_j) \chi * K_2^{*P_{j'}}} +  \abs{K_1^{*P_j}* \chi K_2(t_j) (1-\chi) * K_2^{*P_{j'}}}
  \]
  on $O \times O$, and in particular at $(x,x;t).$ 
  
  In the above, $P_j$ is a partition of some $\tau_j,$ and $P'_j$ a
  partition of some $\tau'_j,$ with $\tau_j + t_j + \tau'_j = t$.
  Eq.~\eqref{eq:kpop-tbd} gives
  $\norm{K_1^{*P_j}}_{\tnm{\tau_j}}\leq e^{b_3 t^{1/2}}$ for some
  $b_3>0.$  
  Thus, provided $y $ is not in $O$, Eq.~\eqref{eq:tpntws-bd}
  ensures
  \[\abs{K_1^{*P_j}(x,y;\tau_j)} \leq
  e^{b_3 t^{1/2}}\sbrace{2(2\pi t)^{-m/2}e^{-r^2/(4t)} + te^{-D^2/(20t)}}
  \leq e^{-c_2/t}\] 
  for some $c_2.$ The bounds in operator norm on $K_1,$ multiplication
  by $\chi,$ and $K_2^{*P_j'}$ readily give
  \[\abs{K_1^{*P_j}* (1-\chi) K_1(t_j) \chi * K_2^{*P_j'}(x,x;t)} \leq e^{-c_3/t}\]
  for some $c_3.$
  The same bound applies  to $\abs{K_1^{*P_j}* \chi K_2(t_j) (1-\chi) * K_2^{*P_j'}},$ so 
  \[\abs{K_1^{*P}(x,x;t) - K_2^{*P}(x,x;t)} \leq \#P \cdot 2e^{-c_3/t}.\]

  Let $P$
  consist of equal intervals with $\#P$ the least integer greater than
  $e^{c_3/(2t)}.$ Thus
  \[\abs{K_1^{*P}(x,x;t) - K_2^{*P}(x,x;t)} \leq 3 e^{-c_3/(2t)}.\]
  On the other hand by Eq.~\eqref{eq:kinf-test} 
	\[
	  \norm{K_i^\infty -
 K_i^{*P}}_{\tnm{t}} \leq A t^{5/4} \abs{P}^{1/4} \leq A t^{3/2} e^{-c_3/(8t)}
	  \]
	  and therefore by Eq.~\eqref{eq:tinf-bd}
	\[\abs{K_i^\infty(x,x;t)- K_i^{*P}(x,x;t)} \leq e^{-d/t}\]
	for some $c,d$ and small enough $t$. Eq.~\eqref{eq:local} then
        follows.
\end{proof}

              \subsection{The heat kernel}\label{ss:heat-kernel}

              Notice that an approximate heat kernel is an approximate
              semigroup by Prop.~\ref{pr:approx-heat}, and thus by
              Thm.~\ref{th:kinf} has a fine partition
              limit. Thm.~\ref{th:heat-kernel} below equates this
              limit with the heat kernel of the same Laplacian as
              appears in the definition of the approximate heat
              kernel.  Thus the fine partition limit offers an
              alternate construction of the heat kernel of a
              generalized Laplacian on a manifold. As noted in the
              introduction, the time-slicing interpretation of the
              path integral depends on a choice of kernel (reflecting
              a discretization of the action) 
              which in some sense approximates the heat kernel
               for a generalized Laplacian quantizing the
              Hamiltonian.
              If this choice is in fact an approximate heat kernel in
              the precise sense of Def.~\ref{def:ahk}, then
              Thm.~\ref{th:heat-kernel} provides a rigorous
              construction of the heat kernel as a time-sliced path
              integral with the appropriate
              Lagrangian. Section~\ref{sc:dirac} spells this out in
              some detail. 

 \begin{lemma} \label{lm:heat-op}
 Suppose $K(t)$ is an approximate heat kernel for the  operator
 $\Delta$ as in Defs.~\ref{def:tame} and~\ref{def:approx-ker}  and $K^\infty(t)$ is the
 limit guaranteed by  
 Thm.~\ref{th:kinf}.
 If $f$ is a smooth 
 section of $\VV$ bounded in each coordinate patch and $t<T,$ 
 \be
 f(t) = K^{\infty}(t)  *   f
 \ee{heat-op}
 agrees with the unique solution  $f(t)$ of the heat equation
 $\partial f(t)/\partial t =\frac{1}{2} \Delta f(t)$ subject to  $\lim_{t \to 0} f(t)=f.$
 \end{lemma} 

 \begin{proof}
   For $f(t)$ given by Eq.~\eqref{eq:heat-op},  Eq.~\eqref{eq:kinf-test}
   of Thm.~\ref{th:kinf}
and  Eq.~\eqref{eq:top-bd} combine to give
   \[\norm{f(t)- K(t) *  f}_\infty \leq 
   At^{3/2}e^{2B\sqrt{t}}\norm{f}_\infty\]
   for $t<T.$  
   In particular 
   \[\lim_{t \to 0} f(t)= \lim_{t\to 0} K(t) *  f= f.\]
  
 To see that the heat equation holds, note $K^\infty$ is a semigroup:
 $K^\infty(t)=K^\infty(t_1) *  K^\infty(t_2),$  for $t = t_1 + t_2$
 and $t_1, t_2 > 0;$ this follows from considering the limit of
 $K(t_1) * K(t_2)$ under refinements of the partition $(t_1,t_2)$ of
 $t.$
 Thus
   \begin{align*}
     \abs{\frac{\partial f(t)}{\partial t} - \frac{1}{2}\Delta f(t)} &= \abs{\lim_{\tau \to 0}
     \frac{f(t+\tau)-f(t)}{\tau} - \frac{1}{2}\Delta f(t)}= \abs{\lim_{\tau \to 0}
     \frac{K^\infty(\tau) *  f(t)-f(t)}{\tau} - \frac{1}{2}\Delta f(t)}\\
     &\leq \abs{\lim_{\tau \to 0}
     \frac{K(\tau) *  f(t) -f(t)}{\tau} -
     \frac{1}{2}\Delta f(t)} +
       \lim_{\tau \to 0} \frac{A\tau^{3/2}e^{2B\sqrt{\tau}}\norm{f(t)}_\infty}{\tau}\\
 & \leq \abs{\lim_{\tau \to 0}
     \frac{K(\tau) *  f(t) -f(t)}{\tau} -  \frac{1}{2}\Delta f(t)} = 0,
   \end{align*}
     where the first line follows from  the semigroup property of $K^\infty,$ the
     second from the preceding estimate, and
     the last from Eq.~\eqref{eq:kf} of
     Def.~\ref{def:ahk}.

 \end{proof}

 \begin{theorem} \label{th:heat-kernel}
 Suppose $K$ is an approximate heat kernel for
  elliptic $\Delta$  on a bundle $\VV$  over $M$ with  Riemannian
 $g.$    Then the fine partition 
 limit $K^\infty(t)$ defined by  Thm.~\ref{th:kinf}
 is
 the heat kernel of $\Delta.$
 \end{theorem}

 \begin{proof}
 Thm.~\ref{th:kinf}
 implies that the refinement limit
   $K^\infty(t)$   exists for all
   $t>0$ and satisfies Eq.~\eqref{eq:kinf-est} for sufficiently small
   $t.$  
  Lemma~\ref{lm:heat-op}  implies as a distribution  $K^\infty$  is a
 solution to the heat equation for $t<T.$   If $t$ is
 too large to apply this lemma directly,
 note that  $K^\infty(t)=\parens{K^\infty}^{* Q_t}$ for some partition $Q_t,$ with
   each $t_i<T.$  Thus, $K^\infty$ is a distributional
   heat kernel for all $t>0.$  Since $\Delta$ is
 elliptic,  elliptic regularity \cite{Evans98}
 says $K^\infty(x,y;t)$  is smooth in $x,$ $y,$ and $t$ and
  thus is the heat kernel of $\Delta.$ 
 \end{proof}

\section{Kernels for Generalized Laplacians and $N=1/2$ SUSY}
\label{sc:dirac}
The results of
Sect.~\ref{ss:heat-kernel} ensure the products of a  kernel which satisfies the conditions
defining an approximate heat kernel will in fact converge to the heat
kernel. This begs the question of how to find such a kernel for a
given Laplacian. Eq.~\eqref{eq:kgen-def} answers this by defining a specific kernel for
each generalized Laplacian. Thm.~\ref{th:kgen} 
applies the results
of Sect.~\ref{ss:heat-kernel}
 to show the fine-partition limit of products
of this
kernel is the heat kernel for that Laplacian. 

Sect.~\ref{ss:dirac} interprets the 
approximate heat kernels for these generalized Laplacians  as exponentiated, discrete actions for an associated supersymmetric
theory. As such these approximate heat kernels provide  the basis for
a time-slicing approximation to the path 
integral for this theory. Thus Thm.~\ref{th:kgen} 
provides a rigorous realization of the 
time-slicing construction of the path integral and
confirms it represents the heat kernel.

Sect.~\ref{ss:susy} specializes this to 
twisted $N=1/2$ supersymmetric quantum mechanics 
in the imaginary-time formulation.
Note that an appropriate choice of twisting (the Levi-Civita
connection on the dual spinor bundle) yields $N=1$ supersymmetric
quantum mechanics.

\subsection{Approximate heat kernel for elliptic operators}

	Let $\VV$ be a vector bundle over a manifold $M$ with
        Riemannian metric $g.$    Berline, Getzler and 
	Vergne \cite{BGV04} observe that every generalized Laplacian can be
	written locally as
	\be
	\Delta^\VV = g^{ij} \sbrace{\nabla_{\partial_i}^\VV
          \nabla_{\partial_j}^\VV -
          \Gamma_{ij}^k\nabla_{\partial_k}^\VV} -V, 
	\ee{delta-def}
	where $\nabla^\VV$ is a connection on $\VV,$
	$\nabla^{\text{LC}}_{\partial_i}(\partial_j)=
	\Gamma_{ij}^k \partial_k$ defines the Christoffel symbols for the
	Levi-Civita connection on the tangent bundle, and $V$ is a section of
	$\End(\VV).$  If $(\VV,M,g, \Delta)$ is tame,
        then, for $d(x,y)<D,$  let 
	$\pt^y_x \in \Hom\parens{\VV_y,\VV_x}$ denote the parallel
        transport map from $\VV_y$ to $\VV_x$ along 
	the unique minimal geodesic. Define the section of $\Hom_{yx}$ 
	\be
	K_\Delta(x,y;t)= H_D(x,y;t)e^{ -\Ricci\parens{\vec{x}_y,
            \vec{x}_y}/12 - \scalar t/12 - tV(x)/2} \pt^y_x,
	\ee{kgen-def}
where the Ricci and scalar curvatures are evaluated at $y$.
	
The following lemma provides some basic estimates on the effect of
having modified $H_D$ of Eq.~\eqref{eq:h-def}
 by a known factor.
\begin{lemma} \label{lm:time-space}
	For any $k \in \NN$ 
    \be 
    d(x,y)^k H_D(x,y;t) \leq 2^{(m+k)/2}(k/e)^{k/2}{t^{k/2}} H_D(x,y;2t).
    \ee{time-space}
Moreover, if $F(x,y;t) = \OO\parens{\abs{\vec{x}_y}^{a}t^{-b}},$
there is a $B > 0$ so that, in the $t$-norm based on this choice,
        $\norm{F(x,y;t)H_D(x,y;t)}_{\tnm{t}} = \OO\parens{t^{a/2-b}}.$ 
If in particular $F(x,y;t)$ is differentiable as a function of
        $y,$ $\vec{x}_y$ and $t,$ and $F(y,y;0)=1,$ then
        there are $B,D$ such that $F(x,y;t)H_D(x,y;t) \in \EE_{B,D}.$   
\end{lemma}


\begin{proof}

    Eq.~\eqref{eq:time-space} follows immediately from the definition
    of $H_D$ in Eq.~\eqref{eq:h-def} and the fact that
    $x^k e^{-x^2/2}$ is bounded by $(k/e)^{k/2}.$

	For  $\abs{F} = \OO\parens{\abs{\vec{x}_y}^a t^{-b}},$
       \[\norm{F(t)H_D(t)}_{\tnm{t}} =
        \norm{\OO\parens{t^{a/2-b}}H_D(2t)}_{\tnm{t}}=\OO\parens{t^{a/2-b}},\] 
where the first equality follows immediately from
Eq.~\eqref{eq:time-space}. 
	
	Similarly, for  $F(x,y;t)= 1 +
        \OO\parens{\abs{\vec{x}_y}} + \OO\parens{t},$ 
	\[F(x,y;t)H_D(x,y;t) = H_D(x,y;t) + \OO\parens{t^{1/2}}H_D(x,y;2t)\]
	and therefore 
\begin{align*}
\abs{F(t)H_D(t)}& \leq H_D(t) + \parens{e^{Bt^{1/2}}-1}H_D(2t) \\
&\leq e^{Bt^{1/2}}\sbrace{e^{-Bt^{1/2}} H_D(t)
  + \parens{1-e^{-Bt^{1/2}}}H_D(2t)} = e^{Bt^{1/2}}\int_1^2 H_D(\alpha
  t)d\mu_\alpha \in \EE_{B,D}(t)
\end{align*}
for small enough $t$ and an appropriate $\mu$ by
Def.~\ref{def:EE}.
\end{proof}
\begin{theorem} \label{th:kgen}
 Suppose   $\nabla^\VV$  is a  connection on $\VV$ a vector bundle over
 a manifold $M,$
 $V$ is a section of $\End(\VV),$  $\Delta$ is 
 the generalized Laplacian associated to this data, and
 $K_\Delta$ is the kernel given by Eq.~\eqref{eq:kgen-def}.   If
 $(\VV,M,g,\Delta)$ is tame, then  
 $K_\Delta$ is an approximate heat kernel
 with constants $(B,C,D)$ depending only on the bounds on $g$ and on
 $\Delta,$ and therefore its large partition limit $K_\Delta^\infty$
 is the heat kernel for $\Delta.$ 
\end{theorem}

\begin{proof} By Thm.~\ref{th:heat-kernel},
 it suffices to verify $K_\Delta$
  satisfies the conditions defining an approximate heat kernel as
  spelled out in Eqs.~\eqref{eq:k2-bd} through~\eqref{eq:ahk}. These
  are all local conditions, so  let $y \in M$ and work  in Riemann
  normal coordinates around $y.$ 
That is, pick an orthonormal basis for $T_{y}M$. Each point $x \in M$ near $y$ is the value of the exponential map at a
unique vector $\vec{x} \in T_{y}M$ near $0$. (The $\vec{x}$ 
 was $\vec{x}_{y}$ earlier; the subscript is implicit here where there is no
danger of confusion.) The components of $\vec{x}$ with respect to the
chosen basis define the Riemann normal coordinates of the point $x$. 
As in the proof of Lemma~\ref{lm:coords}, tameness implies that in Riemann normal
coordinates $g_{ij}$ has bounded $k$th derivatives for $0 \leq k \leq
4.$

If $X$ and $Y$ are tangent vectors at $x \in M$ let $R_x[X,Y]$ be the
Riemannian curvature (endomorphism on $T_xM$), $\Ricci_x(X,Y)$ be the
Ricci curvature, and $\scalar_x$ be the scalar curvature.
The coordinate derivatives $\partial_i$ for $i=1, \ldots,m$  at
each $x \in M$ near $y \in M$  form
 a basis of $T_{x}M$ and define
vector fields in a neighborhood of $y$ 
(commuting but not in general orthonormal). At $y$ these agree with
the original choice of orthonormal basis. Define a second basis
$e_i \in T_{x}M$ (orthonormal but not commuting as vector fields) 
by parallel transporting the same orthonormal basis of
$T_{y}M$ along a minimal geodesic from $y$ to (nearby) $x$.  
The two bases are related by \cite{BGV04}(Prop. 1.28)
\be e_i = \sbrace{\delta_i^j + \frac{1}{6}R_{ikl}^{\rule{1.2em}{0em}j} x^k
  x^l}\partial_j + \OO\parens{\abs{\vec{x}}^3}  
\ee{e-partial} 
where
$R_{ikl}^{\rule{1.2em}{0em}j}\partial_j=R_y[\partial_i,\partial_k]\partial_l$
defines 
the coordinates of the curvature at $y.$
If $g_{ij}\parens{\vec{x}}=
\parens{\partial_i, \partial_j}_x,$
 with
inverse $g^{ij}\parens{\vec{x}},$ and 
 $
  \Gamma_{ij}^k\parens{\vec{x}}\partial_k=
    \nabla_{\partial_i}^{\text{LC}} \partial_j,$
    Eq.~\eqref{eq:e-partial} implies 
\begin{align}
g_{ij}\parens{\vec{x}}&= \delta_{ij} + \frac{1}{3} R_{ikjl} x^k x^l
+ \OO\parens{\abs{\vec{x}}^3} \label{eq:g-rnest}\\
g^{ij}\parens{\vec{x}}&= \delta^{ij} - \frac{1}{3}
R_{k \rule{.2em}{0em} l}^{\rule{.3em}{0em} i \rule{.2em}{0em} j}  x^k x^l 
+ \OO\parens{\abs{\vec{x}}^3} \label{eq:ginv-rnest}\\
\Gamma_{ij}^k\parens{\vec{x}}&= - \frac{1}{3}
\sbrace{R_{ilj}^{ \rule{.8em}{0em} k} + R_{jli}^{\rule{.8em}{0em} k}}x^l
+ \OO\parens{\abs{\vec{x}}^2} \label{eq:gamma-rnest}\\
{\det}^{1/2}g(\vec{x})&= 1 + \frac{1}{6}
R_{ikj}^{\rule{1.1em}{0em}k}x^i x^j +
\OO\parens{\abs{\vec{x}}^3} \label{eq:dmeas-rnest}
 \end{align}
 freely raising and lowering indices using 
$g_{ij}(0)=\delta_{ij}.$  
At $y, $ abbreviate
$\Ricci_{ij}=
R_{ikj}^{\rule{1.1em}{0em} k}=\Ricci_y\parens{\partial_i,\partial_j}$ and 
$\scalar=\Ricci_i^i=\scalar_y.$ 

The bounds implicit in $\OO\parens{\abs{\vec{x}}^p}$ above depend only the
bounds on $g_{ij}$ and its derivatives up to order three. 
Trivialize
the bundle $\VV$ in a ball of radius $D$ around  $y$ by identifying
$\VV_x$ with $\VV_y$ via parallel transport along the unique minimal
geodesic connecting $y$ and $x,$ so that 
\be \nabla_i^\VV= \partial_i + \frac{1}{2}x^jF_{ij}^\VV + \OO\parens{\abs{\vec{x}}^2}
\ee{nabla-rnest}
where $F_{ij}$ is the curvature of $\nabla^\VV$ evaluated at $y$ in
the $\partial_i \wedge \partial_j$ direction
\cite{BGV04}(Prop. 1.18), the bound depending on the bound on the
coefficients of $\nabla$ to order $2.$ 

 Eq.~\eqref{eq:kgen-def} and Lemma~\ref{lm:time-space} give
 $K_\Delta(t) \in \EE_{B,D}$
(Def.~\ref{def:EE})
for some $B > 0.$ Thus, 
$\norm{K_\Delta(t)}_{\tnm{t}}\leq 1,$ verifying  Eq.~\eqref{eq:k2-bd}
of the definition of an approximate heat kernel. 

Using Eqs.~\eqref{eq:g-rnest}-\eqref{eq:nabla-rnest} and the
antisymmetry of $F^\VV,$
\begin{align*}
\Delta&=g^{ij} \sbrace{\nabla_{i}^\VV \nabla_{j}^\VV
  - \Gamma_{ij}^k\nabla_{k}^\VV} - V\\
&=g^{ij}\sbrace{\partial_i \partial_j + \frac{1}{2} F^\VV_{ji}+
  \frac{1}{2}x^k\parens{F^\VV_{ik}\partial_j+F^\VV_{jk}\partial_i} -
  \Gamma_{ij}^k \partial_k +\OO\parens{\abs{\vec{x}}} +\OO\parens{\abs{\vec{x}}^2} \partial_i} - V\\ 
&= \partial^i \partial_i -\frac{1}{3}R_{k \rule{.2em}{0em} l}^{\rule{.3em}{0em} i \rule{.2em}{0em} j} 
x^k x^l\partial_i \partial_j - x^k F^{\VV,i}_{k}\partial_i -V
+\frac{2}{3} \Ricci_i^jx^i\partial_j +  \OO\parens{\abs{\vec{x}}} +
\OO\parens{\abs{\vec{x}}^2} \partial_i +
\OO\parens{\abs{\vec{x}}^3}\partial_i \partial_j. 
\end{align*}

Compute
\[\frac{\partial }{\partial t} K_\Delta(x,y;t) = \sbrace{ -\frac{m}{2t} + \frac{\abs{\vec{x}}^2}{2t^2} - \frac{\scalar}{12} - \frac{V}{2}}K_\Delta(x,y;t),\]
\[\partial_{i,x} K_\Delta(x,y;t) = \sbrace{-\frac{x_i}{t} -\frac{\Ricci_{ij}x^j}{6} + \OO(t)}K_\Delta(x,y;t),\] 
\[\partial_{i,x} \partial_{j,x}K_\Delta(x,y;t) =
\sbrace{-\frac{\delta_{ij}}{t} - \frac{\Ricci_{ij}}{6}+ \frac{x_i
    x_j}{t^2} +\frac{x^i\Ricci_{jk}x^k+x^j\Ricci_{ik}x^k }{6t} + \OO(t
  + \abs{\vec{x}}^2)}K_\Delta(x,y;t),\]  
so
\begin{align*}
  & \sbrace{\frac{\partial}{\partial t} - \frac{1}{2}\Delta}K_\Delta
  =\Big[  -\frac{m}{2t} + \frac{\abs{\vec{x}}^2}{2t^2} -
  \frac{\scalar}{12} - \frac{V}{2}  + \frac{m}{2t} +
  \frac{\scalar}{12} -\frac{\abs{\vec{x}_y}^2}{2t^2} -
  \frac{x^i\Ricci_{ij}x^j}{6t} \\ 
&\qquad  - \frac{x^i\Ricci_{ij} x^j}{6t} + \frac{1}{6t^2}R_{kilj} x^k x^lx^i
x^j + \frac{1}{2t} x^kF^\VV_{ik} x^i + \frac{V}{2} + \frac{x^i\Ricci_{ij}x^j}{3t}\\
&\qquad + \OO\parens{\abs{\vec{x}}+\abs{\vec{x}}^3/t +
  \abs{\vec{x}}^5/t^2+ t}\Big]K_\Delta(x,y;t) \\ 
&=\OO\parens{\abs{\vec{x}}+\abs{\vec{x}}^3/t + \abs{\vec{x}}^5/t^2+ t}K_\Delta(x,y;t)
\end{align*}
after taking into account the antisymmetry of $F^\VV$ and the fourfold
symmetry of $R.$ By Lemma~\ref{lm:time-space}, the right-hand side has
$t$-norm bounded by a multiple of $t^{1/2},$ so the first line of
Eq.~\eqref{eq:ahk} holds.  Since the Laplace-Beltrami operator is
self-adjoint, $\Delta^*$ is the operator associated to $g,$
$\nabla^\dagger$ and $V^\dagger,$ where $\dagger$ represents the
canonical map sending $\End(\RR^n)$ to $\End\parens{(\RR^n)^*}.$ So
for the second line of Eq.~\eqref{eq:ahk} it suffices to observe that
$K_{\Delta^*}(x,y;t)= K_\Delta^\dagger(y,x;t) + \OO(\abs{\vec{x}_y}^3
+ \abs{\vec{x}_y}t)$.
This estimate follows from the tameness
assumption which more directly implies
$\Ricci_x(\vec{y}_x,\vec{y}_x)- \Ricci_y(\vec{x}_y,\vec{x}_y)=
\OO\parens{\abs{\vec{x}_y}^3},$ $\scalar_x - \scalar_y =
\OO\parens{\abs{\vec{x}_y}},$ and $V(y) - \parens{\pt_x^y}^{-1}V(x)
\pt_x^y= \OO\parens{\abs{\vec{x}_y}}$
with the bounds depending on the bounds on the
metric. Eq.~\eqref{eq:ahk} now follows.

For Eq.~\eqref{eq:tto0}, let $f$ be a smooth function on $O$ valued
in $\RR^n.$   Then working in Riemann normal coordinates around $x$
with the the bundle trivialized by parallel transport in radial
directions,  
\begin{align*}
	\lim_{t \to 0} \int K_{\Delta}(x,y;t) \cdot f(y) \dmeas y &=
        \int H_D(x,y;t) f(y) \sbrace{1 + \OO\parens{\abs{\vec{y}_x}^2}
          + \OO(t)} \dmeas \vec{y}_x\\ 
&= f(x) + \OO(t).
\end{align*}
Similarly, for Eq.~\eqref{eq:kf} it suffices by 
the Mean Value Theorem 
to show $\lim_{t \to 0}\frac{\partial}{\partial
  t}K_\Delta*f = \frac{1}{2}\Delta f.$ In Riemann normal coordinates
\begin{align*}
\lim_{t \to 0}\frac{\partial}{\partial t} K_\Delta*f (x) &= \lim_{t \to
  0}\int \sbrace{ -\frac{m}{2t} + \frac{\abs{\vec{x}_y}^2}{2t^2} -
  \frac{\scalar}{12} - \frac{V}{2}}K_\Delta(x,y;t) f(y)\dmeas y\\ 
&= \lim_{t \to 0}\frac{1}{2}\sbrace{\partial_i \partial_i -V}f(x) + \OO\parens{t^{1/2}} = \frac{1}{2} \Delta f
\end{align*}
by straightforward Gaussian integrals.  Finally
Eq.~\eqref{eq:kprime-bd} follows for appropriate $B$ from the above
calculation for $\partial_i K_\Delta.$ 
\end{proof}

\begin{remark} \label{rm:correction}The calculations verifying Eq.~\eqref{eq:ahk} shed some
  light on the role of 
  the  Ricci and scalar curvature terms in the definition of
  $K_\Delta.$ 
Adding $a\scalar t +
b \parens{\Ricci\parens{\vec{x}_y,\vec{x}_y}-\scalar t}$ to the
  exponent in  $K_\Delta$ changes  Eq.~\eqref{eq:ahk} in two ways: it
  would add
  $a\scalar$  to the operator $\Delta$, and  
    $\norm{b \parens{\Ricci\parens{\vec{x}_y,\vec{x}_y}/t-\scalar}}_{(t)}
    \in \OO\parens{1} $
    to the bound  $Ct^{1/2}$ on the right-hand side.
  In units where $\hbar$ is not $1,$ the addition to
 $\Delta$ is $a\scalar \hbar^2$ and thus is a
  quantum correction to the Hamiltonian which presumably corresponds 
  to a different resolution of the operator ordering ambiguity in
  $g_{ij}p^ip^j$.  
Although the
  $\Ricci\parens{\vec{x}_y,\vec{x}_y}/t-\scalar$ term is of too large an order
  in $t$ for Eq.~\eqref{eq:ahk} to hold, it surprisingly does not
  change the fine-partition limit. 
  However, the convergence argument in this
  paper would not suffice in that case.
\end{remark}

\subsection{Path integrals} \label{ss:dirac}
The previous subsection argued that the heat kernel for any
generalized Laplacian $\Delta$ 
can be expressed as a fine-partition limit of
products of
 an
approximate heat kernel constructed directly from $\Delta$. 
As noted in the introduction, the product
 associated to a partition can be viewed as an integral
 over all elements of a discretized space of paths.
Formally, the
limit can be interpreted as an integral over all paths of a function
on the space of paths.  However, it is not obvious that this function
is necessarily 
the exponential of the integral  of a classical
Lagrangian. The generalized Laplacian relevant to the path integral
proof of the index theorem for the twisted Dirac operator cannot be
the quantization of some classical Hamiltonian, because there is not even any symplectic space
on which such a classical Hamiltonian could be defined. Thus, there is
no classical Lagrangian with which to begin formulating a path
integral, even heuristically. 
However, for this case, Friedan and
Windey~\cite{FW84} suggest a natural  extension of the generalized
Laplacian to a larger
space where the 
heat kernel can be written as an integral over all paths of the exponentiated
integral of a Lagrangian, and such that a natural restriction of the
quantum state space recovers the heat kernel for the
original Laplacian.
Sect.~\ref{ss:susy} shows
this trick is unnecessary in the case of untwisted $N=1/2$ supersymmetric
quantum mechanics and, for the twisted 
 $N=1/2$ theory, 
 is only necessary to deal with the
twisted portion. 
Interestingly, the perturbative approximation for the restricted
operator is the same as the restriction of the perturbative
approximation for the unrestricted operator, so either provides an
interpretation of the path integral proof of the index theorem. 

If $f(v_1, \ldots,v_n)$ is a multilinear function of $\VV^*$ for some
vector space $\VV,$ then the antisymmetrization of $f$ represents an
element of $\Lambda \VV.$  To  say $\psi$ is a \emph{Grassman
  variable} valued in $\VV^*,$ means that the expression $f(\psi,
\ldots, \psi) $ represents that element.   If $\VV$ has an inner product
the \emph{Berezin Integral} $\bint f(\psi) 
 d\psi$  is the coefficient  
 of the canonical top-degree  element of $\Lambda \VV$ in $f(\psi)$.
 The inner product induces a nondegenerate pairing on elements of $\Lambda
 \VV$  which in this language becomes
\[\parens{f(\psi), g(\psi)}= \bint f(\psi)g(\psi) d\psi .\]
 See~\cite{MQ86} for a standard reference on Grassman
variables; \cite{Rogers92a} and~\cite{FS14} give examples relevant to
SUSYQM. 

Suppose $(M,g, \VV,\nabla, V)$ are as in Eq.~\eqref{eq:kgen-def}.  Let
$\XX=\Lambda \VV$, and let $\Delta^{ \XX}$ be the generalized
Laplacian associated  to $\nabla$ and $V$ promoted to a
connection and operator on $ \XX$ (using $\nabla(a \wedge b)= \nabla(a)
\wedge b +  a \wedge \nabla(b)$ and $V(a \wedge b)= V(a)
\wedge b +  a \wedge V(b)$). For each point $x \in M$ let $\psi_x$ be
a Grassman variable valued in $\VV^*_x$ so as to
 write kernels on $\XX$ as
superkernels $K(x,y,\psi_x,\psi_y),$ with the understanding $K$ acts on
a section of $\XX,$ which is represented by a  superfunction
$f(x,\psi_x),$ as
\[(K*f)(x,\psi_x)= \int \bint K(x,y,\psi_x,\psi_y) f(y,\psi_y) d\psi_y
\, dy.\]

As outlined in the introduction, given a Lagrangian,  a Riemann sum
approximation to the action defines a kernel which, after some 
 corrections  (higher-order in $\hbar$)                                  
defines an approximate kernel.  Let $\sigma(s)$ be a path in
$M$, let $\Psi,\Psi^\dagger$ be  Grassman variables valued in lifts of
$\sigma$ to $\VV^*$ and $\VV$   respectively, and consider the action
\[\int \frac{1}{2}\parens{\dot{\sigma},\dot{\sigma}} + i
\bracket{\Psi ^\dagger, \nabla^\VV_s \Psi} 
-\frac{i}{2}\bracket{V \Psi^\dagger,  \Psi} ds.\]

On a small interval of
parameter length $t$, approximating the path connecting $x$ and $y$ by a
geodesic gives $\int \frac{1}{2}\parens{\dot{\sigma},\dot{\sigma}} dt \sim
 \parens{\vec{x}_y/ t,\vec{x}_y/ t} t /2 \sim
\abs{\vec{x}_y}^2/(2 t),$
 which agrees with the exponent in $H_D.$
Assuming 
$\Psi^\dagger$ and $\nabla_s\Psi$ are covariantly slowly varying,
$ \int i\bracket{\Psi^\dagger, \nabla_s \Psi} ds \sim i \bracket{
  \Psi^\dagger(t_y) , \pt_y^x \Psi(t_x) -
  \Psi(t_y)}=i\bracket{\psi^\dagger_y, \pt^x_y \psi_x - \psi_y}$
and  $ \int
i\bracket{V\Psi^\dagger, \Psi} ds \sim i \bracket{
  \Psi^\dagger(t_y) ,  t V^*(y)  \Psi(t_y)} \sim i \bracket{
  \psi^\dagger_y ,  t   \pt_y^x V^*(x)\psi_x}$. 
This suggests an approximate heat kernel
 \be K_{\Delta^\XX}(x,y,\psi_x,\psi_y;t) = \bint H_D(x,y;t)
e^{-\Ricci\parens{\vec{x}_y,\vec{x}_y}/12 - t \scalar/12 + i
\bracket{\psi_y^\dagger,  \pt_y^x \sbrace{1-t V^*(x)/2}\psi_x -  \psi_y}}
d\psi_y^\dagger.
\ee{super-k}
The Ricci and scalar
curvature terms do not follow directly from the approximation to the
action. Rather, referring to Rem.~\ref{rm:correction}, they correspond to the resolution of the
operator-ordering ambiguity that gives $\Delta^\XX$ as the operator
whose kernel is the path integral with this Lagrangian, and, among
such choices, they are of the particular form to make
$K_{\Delta^\XX}$ an approximate heat kernel for $\Delta^\XX$. Indeed, 
\begin{proposition} \label{pr:pathint}
  $K_{\Delta^\XX}$ is an approximate heat kernel for $\Delta^\XX$ and
  therefore its fine-partition limit $K_{\Delta^\XX}^{\infty}$ is the
  heat kernel.  Furthermore the component of $K_{\Delta^\XX}^\infty$
  of degree $1$ in $\psi_x$ and degree $\dim \VV -1$ in $\psi_y$ is the
  heat kernel for $\Delta^\VV.$ 
\end{proposition}

\begin{proof}  
 If $A \colon W_x \to
  W_y$ is a linear  map between vector spaces of the same even
  dimension, $\psi_x,\psi_y$ are Grassman variables taking values in
  $W_x$ and $W_y$ respectively, and  $\psi_y^\dagger$ takes values in
  $W_y^*,$ then  for $f$ a function defined on $W_y$,
\[\bint \bint e^{i \bracket{\psi_y^\dagger, A \psi_x - \psi_y}}
f(\psi_y) d\psi_y^\dagger d\psi_y = f(A\psi_x).\]
This is an immediate consequence of the definitions; the authors'
earlier paper spells out the details for the special case $A =
\pt^x_y$~\cite{FS08}. Thus, the quantity
\[\bint e^{i\bracket{\psi^\dagger_y, \pt^x_y \sbrace{1-t V^*(x)/2}\psi_x-  \psi_y}}
d\psi_y^\dagger\]
is, up to terms in $\OO\parens{t^2}$,  the superkernel for the
 operator
\[e^{-tV(x)/2} \pt^y_x \colon \XX_y \to \XX_x,\] which is
the 
extension 
of $e^{-tV(x)/2} \pt^y_x : \VV_y \to \VV_x$.
Therefore Eq.~\eqref{eq:super-k} differs from
Eq.~\eqref{eq:kgen-def} by $\OO(t^2)K_{\Delta^\VV},$ which means that it
also defines an approximate heat kernel for $\Delta^\VV.$ 
\end{proof}

\subsection{Twisted N=1/2 SUSYQM}\label{ss:susy}

\subsubsection{The generalized Laplacian}
The heuristic path integral for twisted $N=1/2$ SUSYQM 
in imaginary time
is supposed to be related to the kernel of the heat operator for a
Laplacian which is the square of the twisted Dirac
operator~\cite{FW84}. 
To define this operator,
 recall some Clifford algebra facts and terminology
detailed in Ch. 3 of~\cite{BGV04}.
  If $M$ is a Riemannian manifold
define $\Clifford=C(T^*M)$ to be the bundle which at each point $x\in M$  is the
complexified $\ZZ/2\ZZ$-graded (and $\ZZ$-filtrated) algebra generated by $T_x^*M,$ subject to the relation
\be
v^*\cdot  w^*+w^*\cdot v^*= -2\parens{v^*,w^*}.
\ee{clifford-def}
A \emph{Clifford module} is a graded  vector bundle $\VV$ over $M$
with a graded homomorphism
$c_\VV\colon\Clifford \to \End\parens{\VV}$.   $\Lambda\parens{T^*M}$ is a Clifford
module with the action $c_\Lambda(v^*)\alpha= v^* \wedge \alpha - i_v(\alpha)$
where $v$ is dual to $v^*$ in the inner product. 

If $M$ is even-dimensional and spin, the spinor
bundle $\spinor= \Lambda \PP,$ where $\PP$ is a polarization of the complexified cotangent bundle of
$M$ is  a Clifford module.
Indeed, with this action,  $\Clifford \iso \End\parens{\spinor},$ and any Clifford
module can be written as $\VV= \spinor \tensor \tspinor,$ where
$\tspinor$ is a vector bundle on which $\Clifford$ acts trivially.

If $\VV$ is a Clifford module, a connection $\nabla^\VV$ is a \emph{Clifford
  connection} if, for any vector field $X$ and 
section $Y$ of $T^*M,$ 
\be
\sbrace{\nabla^\VV_X,c_\VV(Y)}=c_\VV\parens{\nabla_X^{\text{LC}}Y}.
\ee{cliffconn-def}
(The bracket on the left-hand side is graded.)
In the case where $M$ is even-dimensional and spin,  
any  Clifford connection  $\nabla^\VV$  can be written as 
\be
\nabla^\VV= \nabla^\spinor \tensor 1 + 1 \tensor \nabla^\tspinor
\ee{nabla-decomp}
for some connection $\nabla^\tspinor$ on $\tspinor$ and the Levi-Civita
connection $\nabla^\spinor$ on $\spinor$.  If $M$ is even-dimensional
but not spin, the Clifford action is still faithful and the 
curvature of a Clifford connection still decomposes as $R+F^\tspinor,$ 
where $R$ is  Riemannian curvature and $F^\tspinor$ is the component of the
curvature in $\End_{C(M)}(\VV)$. \cite{BGV04}(Props.
3.35,3.40  \& 3.43). 

If $\VV$ is a Clifford module and $\nabla^\VV$ a Clifford connection, the twisted  Dirac operator is
\be
\dirac^\VV= c_\VV(dx^i)\nabla^\VV_{\partial_i}.
\ee{dirac-def}
This provides a square root of the generalized Laplacian $\Delta^\VV$ with the choice of section $V =
c_\VV\parens{F^\tspinor} - \scalar/4$, where  $c_\VV$ acts on
two-forms by $c_\VV(v^* \wedge w^*) = \frac{1}{2}\sbrace{c_\VV(v^*)c_\VV( w^*) - c_\VV( w^*)c_\VV(v^*)}.$ 
That is, with this $V$,
\be
\Delta^\VV = \parens{\dirac^\VV}^2.
\ee{dirac-squared}
(In the special case $\VV = \spinor$, the operator $\dirac^\VV$ is
the ordinary Dirac operator.)
If $(\VV,M,g,\Delta)$ is tame then 
the kernel $K_{\Delta^\VV}$ associated by Eq.~\eqref{eq:kgen-def} to this
data is an approximate heat kernel and converges to the heat kernel of
the square of the Dirac operator by Thm.~\ref{th:heat-kernel}.

\subsubsection{The action}
If $M$ is even-dimensional and spin 
and $\tspinor$ is a bundle over $M$ with a connection whose curvature is
$F$, define  \emph{twisted $N=1/2$ SUSYQM} 
via the
action
\[
\int \frac{1}{2}\parens{\dot{\sigma},\dot{\sigma}} + i
 \bracket{\Psi^\dagger, \nabla^\spinor_s \Psi} + i
 \bracket{\Pi^\dagger, \nabla^\tspinor_s \Pi}- 
\frac{i}{2}\bracket{ F(\Psi, \Psi) \Pi^\dagger,\Pi} ds,
\] 
for $\Psi$ and $\Psi^\dagger$ Grassman-valued lifts of $\sigma$ to $\PP^*$ and $\PP$
respectively, 
and $\Pi$ and
$\Pi^\dagger$ Grassman-valued lifts to $\tspinor^*$ and $\tspinor$
respectively.
This action was first written down by Friedan and  Windey \cite{FW84}
(with slightly different normalization conventions). If $\tspinor$
is the trivial bundle it reduces to the action for 
$N=1/2$ SUSYQM \cite{Alvarez83}.

Discretize as above to get a kernel
on $\hat{\VV}=\spinor \tensor \Lambda \tspinor$
\begin{align}
K_{\mathrm{SUSY}} = &  \bint H_D(x,y;t)
e^{-\Ricci\parens{\vec{x}_y,\vec{x}_y}/12 - t \scalar/12}\label{eq:ksusy} \\
& \times e^{i\bracket{\psi_y^\dagger,  \pt_y^x \psi_x -  \psi_y} + i
\bracket{\eta_y^\dagger,  \pt_y^x \eta_x -  \eta_y} + i t 
\bracket{\eta_y^\dagger,  \pt_y^x\sbrace{F(\psi_x, \psi_x) +  \scalar/4} \eta_x}/2} 
d\eta_y^\dagger d\psi_y^\dagger,  \nonumber
\end{align}
where the parallel transports are with respect to the connections
$\nabla^\spinor$ and $\nabla^\tspinor$. As in the general case, the
terms with parallel transport represent, under Berezin integration,
the kernel of $e^{-\frac{1}{2} \sbrace{c(F) - \scalar/4}}
\pt^x_y$, with this parallel transport being with respect to the
connection 
on $\hat{\VV}$.
Thus the
discretization is exactly the approximate heat kernel
$K_{\Delta^{\hat{\VV}}}$ for  $V = c(F) - \scalar/4$. 

The results of Sect.~\ref{sc:limit} apply to rigorously construct the twisted $N=1/2$
SUSYQM path integral as $K_{\Delta^{\hat{\VV}}}^\infty$ which will
agree 
 with the heat kernel 
for $\Delta^{\hat{\VV}}$.
 Restricting  $K_{\Delta^{\hat{\VV}}}^\infty$
 to the
appropriate degrees in $\eta,$ $\eta^\dagger$ gives the heat kernel
for 
$\Delta^\VV =\parens{\dirac^\VV}^2$ 
where $\VV =\spinor \tensor \tspinor.$ 

Of course the heat kernel for the ordinary Dirac operator is the
fine-partition limit of products of the above without the terms
referring to $\tspinor$. The corresponding action is just
 \[
\int \frac{1}{2}\parens{\dot{\sigma},\dot{\sigma}} + i
\bracket{\Psi^\dagger, \nabla_s^\spinor \Psi} ds, 
\]
which is the usual action for $N=1/2$ SUSYQM.

\section{The Small $t$ Asymptotics} \label{sc:rescale}

 McKean \& Singer~\cite{MS67} recognized that the Gauss-Bonnet-Chern
 theorem would follow from a sufficiently detailed knowledge of the
 short-time diagonal behavior of the heat kernel of the Laplace-de
 Rham 
operator on differential forms. 
They used Duhamel's formula to derive the behavior in degree zero.
  Gilkey~\cite{gilkey84} summarizes an approach of
 Seeley~\cite{seeley67}, 
 Patodi~\cite{Patodi71} and Atiyah, Patodi \&
 Singer~\cite{APS75a,APS75b,APS76} which extends this to cover the
 square of the Dirac operator of Sect.~\ref{ss:susy}, 
 writes the corresponding heat operator as a contour integral, and ultimately
 approximates the heat kernel by approximating the 
 operator in the integrand.  This approximation leads to a heat
 kernel proof of the Atiyah-Singer index theorem.

Witten~\cite{Witten82a,Witten82b} observed that McKean and Singer's
argument fit naturally into the language of supersymmetry, that the
heat kernel for the Dirac operator was the (imaginary-time) propagator
for an appropriate supersymmetric quantum mechanical theory, and that
standard physics calculations of stationary phase/steepest descent
should give the small-time behavior of this
propagator. Alvarez-Gaum\'e 
\cite{Alvarez83} and, independently, Friedan and Windey \cite{FW84}
implemented this program to give path integral ``proofs'' of the index
theorem.  Their arguments differ from earlier heat kernel proofs in
that the  
these small-time asymptotics are computed not from the heat equation directly but from
steepest descent based on the 
Lagrangian appearing in the path integral representation.  Friedan and
Windey in particular cover the general case of 
a Dirac operator associated to an arbitrary Clifford bundle and
Clifford connection, leading to what  Berline, Getzler
and Vergne \cite{BGV04}  refer to as the local index theorem.  The
argument below follows Friedan and Windey closely, although the more 
mathematician-friendly notation and terminology are those of Berline
et al.

Prop.~\ref{pr:local} implies the asymptotics of the heat kernel
at the diagonal are local, so it
suffices to work over $\RR^m$ with a nonstandard metric.
Eq.~\eqref{eq:rescale} rescales the corresponding approximate kernel on $\RR^m$ in a
way  familiar from standard uses of  steepest descent,
 with the extra wrinkle that the
Clifford bundle is also rescaled. The idea is that the rescaling does
not affect the small-$t$
behavior on the diagonal. In fact, Prop.~\ref{pr:kr} shows the
rescaling operation commutes with taking the fine-partition limit. 
On the other hand, Prop.~\ref{pr:rto0}  shows
the rescaled kernel
on a given partition approaches, in a certain limit of the rescaling
parameter, 
 that of a flat
theory with a magnetic term, which is exactly solved in
Prop.~\ref{pr:k0}.  Thm.\ref{th:asit} uses the strong results of
Eq.~\eqref{eq:kinf-test} to interchange the rescaling parameter 
and the fine
partition limits, from which the local version of the the
Atiyah-Singer index
theorem follows directly. 

Suppose $\VV$ is a Clifford 
module over an even-dimensional manifold
$M$ with Riemannian $g,$  $\Nabla^\VV$ is a Clifford connection,
and $V=c_\VV(F^\tspinor)-\scalar/4,$ so the associated elliptic operator
$\Delta^\VV=\parens{\dirac^\VV}^2$ as in Subsection~\ref{ss:susy}.
If 
 all of
that data is tame (for example, compact and smooth), then
$K_{\Delta^\VV}=K_{\dirac^2}$ as 
defined by Thm.~\ref{th:kgen} has a large partition limit
$K_{\dirac^2}^\infty$ which is the heat kernel for $\dirac^2$.

 Let $x_0 \in M$.  Endow a ball of radius $D_1>0$ around $x_0$  with Riemann
 normal coordinates, and identify the bundle over it  with $\VV_{x_0}$
 via parallel transport along minimal geodesics.  This defines a
 metric
 $g_1,$ a trivial bundle $\VV_1,$ and a connection $\Nabla^1$ over a
 neighborhood of the origin in $\RR^m,$ all with bounded derivatives
 up to order four.  Extend all of these to all of $\RR^m$ so that the
 derivatives remain bounded and so that both $\Nabla^1$ and the
 Levi-Civita connection $\Nabla^{g_1}$ continue to be $0$ on radial
 directions. Let $\Clifford$ denote the Clifford algebra  $C(T^*_{x_0}M)$ at
 $x_0 = 0,$
 whose action on $\VV_{x_0}$ splits it into $\spinor
 \tensor \tspinor,$ where $\spinor$  is the spinor representation of
 $\Clifford$ and $\Clifford$ acts trivially on $\tspinor$. $\VV_1$ can be identified with the trivial
 bundle $\spinor \tensor \tspinor$ over  $\RR^m$.
  Identifying the
 Clifford algebra   at any point in $\RR^m$  with $\Clifford$ by
 radial translation gives it an action on $\spinor \tensor
 \tspinor$ that
 makes $\Nabla^1$ a Clifford connection  agreeing 
 with $\Nabla^\VV$
 in the ball of radius $D_1$.  In fact then $\Nabla^1 = \Nabla^{g_1}
 \tensor 1 + 1 \tensor \Nabla^{\tspinor},$ where $\Nabla^{g_1}$ is the
 Levi-Civita connection on $\spinor$ and $\Nabla^\tspinor$ is some
 connection on $\tspinor$ with curvature $F^\tspinor$.  The choice $V_1=c(F^\tspinor)-
 \scalar_1/4$ defines
 a Dirac operator $\dirac_1$ on $(g_1,\spinor\tensor
 \tspinor\times \RR^m, \Nabla^1)$ whose associated approximate heat kernel
 $K_1=K_{(\dirac_1)^2}$
 can be identified with $K_{\dirac^2}$ in that ball
 by the
 obvious isomorphism, and therefore by Prop.~\ref{pr:local}
 \be
K^\infty_{\dirac^2}(x_0,x_0;t)-K^\infty_1(0,0;t) =\OO\parens{e^{-(d_1)^2/t}}
\ee{dirac-local}
for some $d_1 > 0$.

 To investigate the small-time asymptotics of $K_{\dirac^2}$ at $x_0$
 it thus suffices to consider only $K_1(0,0;t)$.  Because the bundle
 is trivial, $K_1$ can be taken not as a section but as a function with
 values in $\End\parens{\spinor} \tensor \End\parens{\tspinor} \sim \Clifford
 \tensor \End(\tspinor)$.
The Clifford algebra
action $c_{\Lambda}$ on $\Lambda T^*_{x_0}M$ means $K_1$ also picks out a function
with values in $\End\parens{\Lambda
  T^*_{x_0}M}\tensor \End(\tspinor)$. Mildly abuse notation to let $K_1$
also refer to this function. 
Thus, $K_1$ is a kernel on the trivial bundle $\Lambda T^*_{x_0}M
 \times \tspinor$ over $\RR^m$, and  $K_1$ is still
 an approximate semigroup with the same constants as before, call them
 $(B,C,D, T)$. 
 
To rescale $K_1,$ define a family of  metrics $g_r$ on $\RR^m$ for $0 \leq r \leq 1$ as follows: Define $\phi_r \colon \RR^m
 \to \RR^m$ by $\phi_r(\vec{x})=r\vec{x}$, define $\psi_r \colon
 \Lambda T^*_{x_0}M \to \Lambda T^*_{x_0}M$ by $\psi_r(\alpha)=
 r^{\deg(\alpha)} \alpha$ for $\alpha$ homogeneous. Finally, define
 $g_r = r^{-2} \phi_r^*[g_1]$. and extend by continuity to
 $g_0=g_{1,\vec{0}}$. By construction, $g_{1,\vec{0}} (\vec{v},
 \vec{w}) = \parens{\vec{v}, \vec{w}}$, 
 the standard inner product on $R^m$. 
  This family has the following properties (extending each
 formula by continuity to $r=
 0$):
 \begin{align*}
	g_{r,\vec{x}}(\vec{v}, \vec{w}) & = g_{1,r\vec{x}}(\vec{v}, \vec{w})  \\
	d_{g_r}(\vec{x},\vec{y}) & =  r^{-1}d_{g_1}(r\vec{x}, r\vec{y} )    \\
	\parens{\vec{y}_\vec{x}}_{g_r}  & =  r^{-1} \parens{(r\vec{y})_{r\vec{x}}}_{g_1} \\
	\Ricci_r(\vec{y}_{\vec{x}},\vec{y}_{\vec{x}}) & =
        \Ricci_1\parens{(r\vec{y})_{r\vec{x}},(r\vec{y})_{r\vec{x}}}
        \\ 
	\scalar_r & = r^2 \scalar_1\\
	\dmeas_{g_r}\vec{y} & = r^{-m} \dmeas_{g_1} (r\vec{y}).
 \end{align*}
 
 If $K(\vec{x},\vec{y};t)$ is a kernel on the bundle $\Lambda T^*_{x_0}M
 \times \tspinor$ over $\RR^m,$ let $\Phi_r$ rescale $K$
 according to 
 \be \Phi_r[K](\vec{x},\vec{y};t) = r^{m} \psi_r^{-1}
 K(r\vec{x},r\vec{y};r^2t) \psi_r.
\ee{rescale}
This rescaling extends $K_1$ to a family of
 kernels $K_r$ on the same bundle via
 \[K_r=\Phi_r(K_1). \]

\begin{proposition}\label{pr:kr}
	$\Phi_r$ is a homomorphism from the kernel product $*$ using
  the metric $g_1$ to the kernel product $*$ using the metric $g_r$
  for $r>0$. A  $t$-norm can be chosen for each
  metric $g_r$ and constants $(B,C,D,T)$ independent of $0<r<1$ such
  that $\Phi_r$ is a map of norm at most $1$ between the respective $t$-norms and
  such that $K_r$ with the metric $g_r$ is an approximate semigroup
  with constants $(B,C,D,T)$ independent of $r.$ Finally $K_r^\infty=
  \Phi_r\sbrace{K_1^\infty}$. 
\end{proposition}

\begin{proof}
	For the first claim
	\begin{align*}
          & \Phi_r[K](t_1)*\Phi_r[J](t_2)(\vec{x},\vec{z}) =
		\int \Phi_r[K](\vec{x},\vec{y};t_1)\Phi_r[J](\vec{y},\vec{z};t_2) 
                d_{g_r}\vec{y}  \\
                & \qquad = \int \psi_r^{-1} r^{2m}
                K(r\vec{x},r\vec{y},r^2 t_1) J(r\vec{y}, r\vec{z};r^2
                t_2) \psi_r d_{g_r}\vec{y}\\ 
		& \qquad = r^{m}  \psi_r^{-1} \int
                K(r\vec{x},\vec{u},r^2 t_1) J(\vec{u}, r\vec{z};r^2
                t_2) d_{g_1} \vec{u} \psi_r\\ 
		& \qquad = r^{m}  \psi_r^{-1} K*J(r\vec{x},r\vec{z};r^2 t) \psi_r=
                \Phi_r[K*J](\vec{x},\vec{z};t).
	\end{align*} 

	For the second, write $B_1,D_1$ for the corresponding
        constants in 
        Prop.~\ref{pr:t-def} and
        Cor.~\ref{cr:t-norm} as determined by the bounds for
        $g_1$. Since the supremum norm on $g_r$ and all its 
        derivatives are bounded by the corresponding quantities for
        $g_1,$ these constants work for any $g_r.$ In particular,
        there is a $t$-norm satisfying Cor.~\ref{cr:t-norm} independent of $r$.
        Write $H_{D,g}$ for the kernel $H_D$
        defined by Eq.~\eqref{eq:h-def} and $K_{B,D,g}$ for the kernel
        $K_{B,D}$ defined by Eq.~\eqref{eq:kbd-def} to emphasize their
        dependence on a given metric. Defining $J$ by
        \[K_{B_1,D_1,g_1}= \chi_{<rD_1}K_{B_1,D_1,g_1} + J,\]
        conclude that 
        \[\Phi_r\sbrace{K_{B_1,D_1,g_1}}=K_{B_1,D_1,g_r} + \Phi_r\sbrace{J}\]
        with $\norm{\Phi_r\sbrace{J}}_{\text{ker}}\leq
        te^{-D_1^2/(20t)}$ for small enough $t$.   
        Suppose $\norm{K}_{\tnm{t}}=  1$, so
       $K = \widetilde{K} + \widetilde{J}$ where
        $\abs{\widetilde{K}} \leq
        e^{B_1 \sqrt{t}} \int K_{B_1,D_1,g_1}(\alpha t) d\mu_\alpha$
        and $\norm{\widetilde{J}}_{\tnm{t}}\leq te^{-D_1^2/(20t)}$. Then  
        \begin{align*}
          \abs{\Phi_r(\widetilde{K})}&\leq e^{B_1 r\sqrt{t}} \int
        \Phi_r\parens{K_{B_1,D_1,g_1}(\alpha t)} d\mu_\alpha\\
        &\leq e^{B_1 \sqrt{t}} \int
        K_{B_1,D_1,g_r}(\alpha t) + \Phi_r\sbrace{J(\alpha t)}  d\mu_\alpha,
        \end{align*}
        so $\Phi_r\parens{\widetilde{K}}$ is an element of
          $\EE_{B_1,D_1,g_r}(t)$ plus a  kernel $J'$ with
        $\norm{J'}_{\text{ker}}\leq 2t e^{B_1 \sqrt{2t}}
          e^{-D_1^2/(40t)}$. 
        Meanwhile
        $\norm{\Phi_r(\widetilde{J})}_{\text{ker}}\leq \norm{\widetilde{J}}_{\text{ker}}
        \leq te^{-D_1^2/(20t)}$.  So replacing $D_1$ with a smaller 
        $D$ makes $\Phi_r(K) \in \EE'_{B_1,D, g_r}(t).$  With
        these choices of constants, which still depend only the bounds
        on $g_1,$ 
        $\Phi_r$ is norm at most $1$ as a map
        between the corresponding $t$-norms.

        Since  $\norm{K_1}_{\tnm{t}}\leq 1$, the preceding argument
        implies   $\norm{K_r}_{\tnm{t}}\leq 
        1$.  Notice also that
        \begin{align*}
        &\norm{K_r(t_1)*K_r(t_2) -
          K_r(t)}_{\tnm{t}}=\norm{\Phi_r\parens{K_1(t_1) *
              K_1(t_2)-K_1(t)}}_{\tnm{t}}\\
          &\qquad \leq \norm{ K_1(t_1) *
            K_1(t_2)-K_1(t)}_{\tnm{t}} \leq Ct^{3/2}
          \end{align*}
        so $K_r$ is an approximate semigroup with constants
        independent of $r$.
       
\end{proof}

 As $r \to 0,$ the rescaled kernel $K_r$ will approach a kernel $K_0$ defined as follows:
 First, define $\Rlimit \in \Lambda T^*_{x_0}M \tensor \End(T_{x_0}M)$ by
 \[\Rlimit_k^l= \frac{1}{2}R_{ijk}^{l}  dx^i \wedge dx^j,\]
 where $R$ is evaluated at $x_0,$ 
 and then define $\Flimit \in \Lambda T^*_{x_0}M \tensor  \End\parens{\tspinor}$ by
 \[\Flimit=\frac{1}{2}F_{ij}^\tspinor
 dx^i \wedge dx^j,\]
 where $F$ is likewise evaluated at $x_0$.
Finally, define
\be
H_{\flat}(\vec{x},\vec{y};t)= (2\pi t)^{-m/2}
e^{-\abs{\vec{y}-\vec{x}}^2/(2t)}
\ee{hlfat-def}
and 
  the kernel on $\Lambda T^*_{x_0} M$ 
 \be
 K_0(\vec{x},\vec{y};t) = H_{\flat}(\vec{x},\vec{y};t)
 e^{\parens{\Rlimit \vec{x}, \vec{y}-\vec{x}}/4 -t\Flimit/2}
 \ee{k0-def}
   where the elements of $\Lambda T^*_{x_0}M$ on the right-hand side 
act by multiplication.  

 \begin{remark}
 These analytic functions of $\Rlimit$ and $\Flimit$ are defined via power series, and  are well-defined for all $t$ because $\Rlimit$ and $\Flimit$ are nilpotent.
 \end{remark}

\begin{lemma}\label{lm:krlim}
\be
\lim_{r\to 0} K_r=K_0
\ee{krlim}
pointwise.	   
\end{lemma}

\begin{proof}
Using $K_r = \Phi_r(K_1),$ the definition of $K_1$ as $K_{(\dirac_1)^2}$ of
Eq.~\eqref{eq:kgen-def} for an elliptic operator $\Delta_1=(\dirac_1)^2$ over
$\VV_1$ 
and the definition of $K_0$ above, the
statement expands to
\begin{align*}
&\lim_{r \to 0} r^m (2\pi t)^{-m/2} e^{-\sbrace{d_{g_r}(\vec{x},
      \vec{y})}^2/(2t)} \\ 
&\qquad \times e^{-\Ricci_r(\vec{y}_\vec{x},\vec{y}_\vec{x})/12  + t
          \scalar_r / 24 - \frac{t}{4} F_{ij}^\tspinor(r\vec{x})\psi_r^{-1}
          r^2 c(dx_i)c(dx_j)\psi_r} \psi_r^{-1}
        \pt_{r\vec{x}}^{r\vec{y}} \psi_r \\
&\qquad \qquad \qquad = H_{\flat} e^{\parens{\Rlimit \vec{x}, \vec{y}-\vec{x}}/4
  -t\Flimit/2} .
\end{align*}
  Thus the lemma 
reduces to the the following assertions:
  \begin{align*}
    \lim_{r\to 0}
    d_{g_r}(\vec{x},\vec{y})&=\abs{\vec{x}-\vec{y}} \\
    \lim_{r\to 0} \Ricci_r(\vec{y}_\vec{x},\vec{y}_\vec{x})&=0,\\
    \lim_{r\to 0} \scalar_r&=0,\\
    \lim_{r\to 0} F_{ij}^\tspinor(r\vec{x})\psi_r^{-1}
          r^2 c(dx_i)c(dx_j)\psi_r/2&=\Flimit \mbox{ and }\\
    \lim_{r\to 0} \psi_r^{-1}
        \pt_{r\vec{x}}^{r\vec{y}} \psi_r&=e^{\parens{\Rlimit \vec{x}, \vec{y}-\vec{x}}/4}.
  \end{align*}
  The first three are immediate from the fact that $g_1= g_0 +
  \OO\parens{\abs{\vec{x}}^2}$.  The fourth follows from the fact that
  $\lim_{r\to 0}\psi_r^{-1}rc(dx)\psi_r=dx$.

  For the fifth limit, having
  trivialized the bundle radially at the origin, the parallel
  transport from $r\vec{x}$ to $r\vec{y}$ is the holonomy of the
  geodesic triangle from $0$ to $r\vec{x}$ to $r\vec{y}$ to $0$.
  Treat 
 this separately on $\tspinor$ and on
  $\Lambda(T^*_{x_0}M)$.  On $\tspinor$ the holonomy differs from $1$
  by a quantity proportional to the area enclosed, which is
  $\OO(r^2)$. 
  For the  $\Lambda(T^*_{x_0}M)$ piece, the holonomy is an element
  of the spin group and therefore an exponential of a degree-two
  element of $\Clifford$. 
  This exponent
  in turn is the image under $c$ of
  the two-form generating
  the holonomy about the same geodesic triangle with respect to the
  Levi-Civita connection. It is standard \cite{AS53}
  that this is $\parens{\Rlimit \cdot r\vec{x}, r\vec{y}-r\vec{x}}/4
  + \OO\parens{\abs{r\vec{x}} \abs{r\vec{y}-r\vec{x}}
    \abs{r\vec{y}+r\vec{x}}}$. Thus, this piece 
  is the exponential of
  the image under $c$ of $\parens{\Rlimit\cdot  r\vec{x},r\vec{y}}/4+
  \OO\parens{r^3}$. Conjugation by $\psi_r$ will
  reduce the power of $r$ by two, giving the result. 
\end{proof}

\begin{proposition}\label{pr:rto0}
	Given any partition $P$ of any $t>0$ 
	\[ \lim_{r \to 0} K_r^{*P}(0,0;t)=K^{*P}_0(0,0;t).\]
\end{proposition}

\begin{proof}
This follows from Lemma~\ref{lm:krlim}, Lebesgue Dominated
Convergence and the fact that $K_0$ and $K_r$ are bounded by $C_1
H(\vec{x},\vec{y};C_2t)$ for some $C_1,C_2$ where
$H(\vec{x},\vec{y};t)=(2\pi t)^{-m/2}
e^{-d_{g_0}^2(\vec{x},\vec{y})/(2t)}$, which in turn follows from the
same bound on $K_1$. 
\end{proof}


\begin{proposition}\label{pr:k0}
\be
	\lim_{\abs{P} \to 0} K_0^{*P}=K^\infty_0
\ee{k0-inf}
	converges pointwise, and is the heat kernel for the operator
\be
	\Delta= \frac{\partial^2}{\partial x_i \partial x_i} +
        \frac{1}{2}\Rlimit_i^j x_j \frac{\partial}{\partial x_i
        }- \Flimit + \abs{\Rlimit \cdot \vec{x}}^2/16.
\ee{k0-delta}
	 Therefore, 
	\be
	K_0^\infty(0,0;t)= (2\pi t)^{-m/2}
{\det}^{1/2}\parens{\frac{t\Rlimit/4}{\sinh(t\Rlimit/4)}} e^{-t\Flimit/2}.
	\ee{k0-form}
\end{proposition}

\begin{proof}

A slight modification of Roger's proof of Theorem 8.2 in \cite{Rogers03} would give this result. 
However, that argument refers to expectations in a
variant of Wiener measure. The following proof uses the
language of products of approximate kernels. 
  First check that for small enough $t=t_1+t_2>0$, 
  \[K_0(t_1)*H_\flat(t_2)= \sbrace{1 + \frac{t_1}{2} \parens{\Delta-\Delta_\flat}}H_\flat(t) +
 \OO\sbrace{\parens{\frac{t_1^2}{t}}\parens{ 1 + \abs{\vec{x}}^m}} H_\flat(2t).
\] 
To see this, begin by directly computing the Gaussian integral, and
use the skew-symmetry of $\Rlimit$ to obtain
  \begin{align*}
    &\sbrace{K_0(t_1)*H_\flat(t_2)}(\vec{x},\vec{z};t)
=(4 \pi^2 t_1 t_2)^{-m/2} \int
    e^{\parens{\vec{x},\Rlimit (\vec{y}-\vec{x})}/4 - t\Flimit/2-
    \abs{\vec{y}-\vec{x}}^2/(2t_1) -
    \abs{\vec{z}-\vec{y}}^2/(2t_2)} d\vec{y} \\
  &\qquad =H_\flat(\vec{x},\vec{z};t)e^{\frac{t_1}{4t}\parens{\vec{x},\Rlimit
      (\vec{z}-\vec{x})} -\frac{t_1}{2} \Flimit + \frac{t_1 t_2}{32 t} \abs{\Rlimit
    \vec{x}}^2}.
\end{align*}
  Notice that since $\Flimit$ and $\Rlimit$ take values in the algebra
  $\Lambda T^*_{x_0}M \tensor \End(\tspinor)$ and are therefore
  nilpotent, the exponential truncates to
  multinomials. 
Expanding the exponential, and comparing  the definition of
$\Delta$ with the Laplacian $\Delta_\flat$ corresponding to the
Euclidean metric on $R^m,$ gives
 \begin{align*}
    &\sbrace{K_0(t_1)*H_\flat(t_2)}(\vec{x},\vec{z};t)
= \sbrace{1 +
     \frac{t_1}{2}\parens{\Delta-\Delta_\flat}}H_\flat(\vec{x},\vec{z};t)
+ \frac{t_1}{32}\abs{\Rlimit \vec{x}}^2\parens{\frac{t_2}{t} -
1}H_\flat(\vec{x},\vec{z};t) \\
& \qquad + \OO\parens{\abs{\frac{t_1}{4t}\parens{\vec{x},\Rlimit
      (\vec{z}-\vec{x})} -\frac{t_1}{2} \Flimit + \frac{t_1 t_2}{32 t}
    \abs{\Rlimit \vec{x}}^2}^2}H_\flat(\vec{x},\vec{z};t).
\end{align*}
Noting $1 - \frac{t_2}{t} = \frac{t_1}{t},$ $\frac{t_1 t_2}{t} \leq
t_1$ and, as in the proof of Lemma~\ref{lm:time-space}, $\abs{\vec{z}-\vec{x}}^k H_\flat(\vec{x},\vec{z};t)$ is bounded by a multiple of
  $H_\flat(\vec{x},\vec{z};2t)$, it is easy to bound the error term by
  $\OO\sbrace{\parens{\frac{t_1^2}{t}}\parens{ 1 + \abs{\vec{x}}^m}}
 H_\flat(\vec{x},\vec{z};2t)$ as claimed. 

  If $P_n$ is the partition
  $(t/n,t/n, \ldots, t/n)$ then
 \begin{align*}
    K_0^{*P_n}&=\sbrace{H_\flat + (K_0-H_\flat)}^{*P_n}
    \\
& =\sum_{k=0}^m  \sum_{\stackrel{i_0+i_1+ \cdots +
    i_k=n-k}{\scriptscriptstyle i_j \geq 0}} H_\flat(i_0 t/n) *
  \sbrace{K_0(t/n)-H_\flat(t/n)} * H_\flat(i_1 t/n) * \\
& \qquad \cdots *
\sbrace{K_0(t/n)-H_\flat(t/n)} *  H_\flat(i_{k-1} t/n) * \sbrace{K_0(t/n)-H_\flat(t/n)} * H_\flat(i_k t/n),
\end{align*}
where the first sum is only to $m$ because $K_0-H_\flat$ is of degree
at least one in $\Lambda T^*_{x_0}M$.

  Replace each of the  $ \sbrace{K_0(t/n)-H_\flat(t/n)}*
  H_\flat(i_j t/n)$ with $\frac{t}{2n}(\Delta-\Delta_\flat)
  H_\flat((i_j + 1) t/n) +
  \OO\sbrace{\frac{t}{n(i_j+1)}\parens{1+\abs{\vec{x}}^m}}H_\flat(2[i_j +
  1]t/n)$. 
The second term introduces  into the
sum a finite number, independent of $n$,  of ``error terms''. Each
contributes a summand 
\begin{align*}
&\OO\sbrace{\frac{t}{n(i_j+1)}\parens{1+\abs{\vec{x}}^m}}
 \sum_{i_0+i_1+ \cdots + i_{k}=n-k} H_\flat(i_0 t/n) *
  \frac{t}{2n}(\Delta-\Delta_\flat)
  H_\flat([i_1 + 1] t/n)* \\
& \qquad \cdots * H_\flat(2 i_jt/n) * \cdots *
\frac{t}{2n}(\Delta-\Delta_\flat) H_\flat([i_k + 1]
t/n)\\
&= \OO\sbrace{\frac{t}{n(i_j+1)}\parens{P[\abs{\vec{x}}]}} H_\flat(2t)
\end{align*} 
                       
where $P$ is some polynomial.  
 As $n$ goes to
infinity, the contribution of each of the finitely-many error terms goes to zero, leaving
\begin{align*}
&K_0^{*P_n}= \sum_{k=0}^m \frac{t^k}{n^k}\sum_{\stackrel{i_0+ i_1 + \cdots + i_k
= n-k}{\scriptscriptstyle i_j \geq 0}}  H_\flat(i_0 t/n) *
  \frac{1}{2}(\Delta- \Delta_\flat) H_\flat([i_1 + 1] t/n)
    *\\
& \qquad \cdots  *\frac{1}{2} (\Delta- \Delta_\flat)
  H_\flat([i_k + 1] t/n)\\
& =  \sum_{k=0}^m \frac{t^k}{n^k}\sum_{\stackrel{i_0+ i'_1 + \cdots + i'_k
= n}{\scriptscriptstyle i_0 \geq 0; i'_j \geq 1}}  H_\flat(i_0 t/n) *
  \frac{1}{2}(\Delta- \Delta_\flat) H_\flat(i'_1 t/n) * \cdots  *\frac{1}{2} (\Delta- \Delta_\flat)
  H_\flat(i'_k t/n)\\
& \stackrel{n \to \infty}{\to} H_\flat(t) + \sum_{k=1}^m
\int \cdots \int_{\stackrel{t_0+t_1+\cdots +t_k=t}{\scriptscriptstyle t_j \geq 0}} H_\flat(t_0) *\frac{1}{2}(\Delta-\Delta_\flat)
H_\flat(t_1) *\\
\qquad  &\cdots  * \frac{1}{2}(\Delta-\Delta_\flat)
H_\flat(t_k)  d t_k \cdots dt_1. 
\end{align*}
This sum agrees with McKean and Singer's expression~\cite{MS67} for the heat kernel
 for $\Delta$, as a sum of $k$-fold $\#$ products, which they derive
from Duhamel's formula.  Their expression has the sum taken over all
nonnegative integers $k$, but, again,
terms with
more than $m$ factors of $\Delta - \Delta_\flat$ vanish due to their
form degrees.

\end{proof}

\begin{lemma} \label{lm:psi}
If $A\in \Clifford_{2k}$, the degree-$2k$ subset in the Clifford
filtration, then the map taking $A$ to $\lim_{r \to 0}
\psi_r^{-1}c_\Lambda(A) \psi_r r^{2k}$ is multiplication by
$\rho_k(A)\in \Lambda^{2k}T^*_{x_0} M$, where $\rho_k(A)$ denotes the
degree $2k$ component of $c_\Lambda(A)1$.  In particular $\rho_k(A)$
and hence the limit is
zero on $\Clifford_{2k-1},$ 
and gives the standard identification of
$\Clifford_{2k}/\Clifford_{2k-1}$ with forms of degree $2k$.
\end{lemma} 

\begin{proof}
Conjugation by $\psi_r$ on $\End\parens{\Lambda T^*_{x_0} M}$
multiplies homogeneous operators of degree $k$ on $\Lambda T^*_{x_0}
M$ by $r^{-2k}$.  Since $c_\Lambda(A)$ is a sum of maps on forms of
homogeneous degrees
ranging from $0$ to $2k$, the small $r$ limit will project onto the
degree $2k$ component.  In fact, this projection acts as
multiplication by $\rho_k(A)$. In particular it is zero on
$\Clifford_{2k-1}$ and sends $v_1^* v_2^* \cdots v_{2k}^*$ to $v_1^*
\wedge v_2^* \wedge \cdots \wedge v_{2k}^*$, proving the second
sentence.
\end{proof}

The following is what Berline, Getzler and Vergne call the local
version of the Atiyah-Singer index theorem for a general Dirac
operator.  The index theorem follows directly from this as in
\cite{BGV04}. 

\begin{theorem} \label{th:asit}
If $\dirac$ is a Dirac operator on a Clifford bundle $\VV$ over a
smooth, compact,  
oriented Riemannian manifold $(M,g)$ of dimension $m,$ and $x_0 \in M,$ then
the diagonal of the heat kernel $K^\infty_{\dirac^2}(x_0,x_0;t)$ of $\dirac^2$
 is asymptotic to a Laurent
series in $t$ of the form \[P(t) = \sum_{k = 0}^{\infty} A_k t^{k
  -m/2},\]
with $A_k \in \Clifford_{2k}\tensor \End(\tspinor)$. Writing $\rho\parens{P(t)}=  \sum_{k = 0}^{m/2} \rho_k(A_k) t^{k
  -m/2}$ for $\rho_k$ as in Lemma~\ref{lm:psi},
\[
\rho\parens{P(t)} =
(2\pi t)^{-m/2} 
{\det}^{1/2}\parens{\frac{t\Rlimit/4}{\sinh(t\Rlimit/4)}} e^{-t\Flimit/2} \label{eq:asit}
\]
where $\Rlimit$ and $\Flimit$ are the curvature forms at $x_0$ as in
the definition of $K_0$ in
Eq.~\ref{eq:k0-def}.
\end{theorem}

\begin{proof}
  By Eq.~\eqref{eq:dirac-local}, it suffices to prove the result for
  $K_1=K_{(\dirac_1)^2}$.
	Using \cite{BGV04}[Thm. 2.30],
        $K^\infty_{1}(0,0;t)= \sum_{i=0}^{(m+2)/2} A_i  t^{i-m/2} + \OO\parens{t}$ 
        for some $A_i \in \Clifford \tensor \End(\tspinor)$.
        Eqs.~\eqref{eq:dirac-def}
        and~\eqref{eq:dirac-squared} imply each $A_i$ is of degree $2i$ in the
        Clifford filtration, so $c_\Lambda(A_i)$ is 
        of degree at most $2i$  as an
        element of $\End(\Lambda T^*_{x_0}M)
        \tensor \End(\tspinor)$.
        
        By Lemma~\ref{lm:psi}, continuing to write $K_1$ for both the  
        kernel and its image under $c_\Lambda$, 
	\[\lim_{r \to 0} \Phi_r\sbrace{K_1^\infty}(0,0;t)=  \lim_{r
          \to 0}  \sum_{i=0}^{m/2}			
        \psi_r^{-1} c_\Lambda(A_i) \psi_r r^{2i}t^{i-m/2} =
        \rho\parens{P(t)}, \]    
since the last term in the above sum for $K^\infty_{1}(0,0;t)$ and the error term get taken to  $\lim_{r \to 0}\psi_r^{-1} c_\Lambda(A) \psi_r r^{m+2}$, which is $0$ for $A \in \Clifford  \tensor \End(\tspinor)$.
       Thus the theorem is equivalent in light of Prop.~\ref{pr:k0} to 
	\[\lim_{r \to 0} \Phi_r[K_1^\infty](0,0;t)= \lim_{\abs{P} \to
          0} K_0^{*P}(0,0;t).\] 
	
	To prove this statement, fix $t>0$.  
        Given $\epsilon$, choose $P$ so that both
	\begin{align*}
		\abs{K_0^{*P}(0,0;t)- \lim_{\abs{P} \to 0} K_0^{*P}(0,0;t)}& <\epsilon/3,\\
		\abs{K_r^{*P}(0,0;t) - \Phi_r[K_1^\infty](0,0;t)} &< \epsilon/3
	\end{align*}
	for all $0<r\leq 1$, where the second estimate follows from Thm.~\ref{th:kinf} and the uniformity of the constants for the $K_r$.
        For that $P$ choose $r$ so small by
        Prop.~\ref{pr:rto0} that  
	\[\abs{K_r^{*P}(0,0;t)-K_0^{*P}(0,0;t)} < \epsilon/3\]
	by the pointwise convergence of $K_r$.  
\end{proof}

\begin{remark}
Roughly speaking, this proof implements the steepest descent
approximation (which is the imaginary-time version of stationary
phase) of the rigorous path integral, for 
the leading terms.
Recall that steepest descent approximates $\int
e^{\phi(x)/\epsilon} dx$ by expanding $\phi$ in a Taylor series about
a critical point, and rescaling $x$ by $\sqrt{\epsilon}$. Choosing to throw away all terms of positive power in $\epsilon$ replaces
$\phi$ by a quadratic approximation $\phi_q$, and the 
approximation to the integral is $\int e^{\phi_q(x)/\epsilon} dx$.
Applying this reasoning heuristically to the path integral, using
$\hbar$ as the parameter, results  in a Lagrangian  with, in general, harmonic
oscillator and linear magnetic terms. Various standard approaches,
including Wiener measure, apply to evaluate this path integral with a
purely quadratic exponent, giving what is
termed the ``semiclassical approximation''.

Lemma~\ref{lm:krlim}, and Prop.~\ref{pr:rto0} give a rigorous version
of this argument, except with $r$ as the small parameter instead of
$\hbar$. Moreover, the rescaling here involves both space and
the Clifford bundle, and the expansion is about the constant path.
Prop.~\ref{pr:k0} rigorously defines the path integral with the
quadratic action by time-slicing and the fine partition limit.  The
interchange of the small-$r$ and fine-partition limits 
concluding the proof of the theorem above thus provides, in
this sense,  a rigorous
proof  of the leading terms of the steepest descent approximation for this nontrivial
path integral.
\end{remark} 

\section{Conclusion}
The argument culminating in Thm.~\ref{th:asit} is a direct translation of the heuristic path
integral proof of the Atiyah-Singer index theorem for the twisted
Dirac operator into rigorous mathematics. Prop.~\ref{pr:pathint}
and Thm.~\ref{th:heat-kernel} provide a rigorous version of the relevant time-sliced path
integral for
each of
 a set of theories including twisted SUSYQM;
as expected, the path integral agrees with the heat kernel. In fact, this
gives a new construction of the heat kernel.
 That the steepest
descent approximation $K_0^\infty$ to the path integral indeed gives its
asymptotic behavior on  the diagonal is the crux of Thm.~\ref{th:asit}. 
The explicit calculation in Prop.~\ref{pr:k0} 
of
$K_0^\infty$ thus gives the asymptotic behavior of the heat kernel and
with it the index theorem.

With the appropriate choice of twisting bundle, namely $\tspinor$
 being the
dual spinor bundle, Thm.~\ref{th:asit} also
gives the Gauss-Bonnet-Chern theorem, which was the subject of the
authors' recent work~\cite{FS14}. Friedan and Windey~\cite{FW84}
prove this by explicitly reducing 
$K_0^\infty$ for this case to the Pfaffian expression of the
Chern form.

Of course the  Laplace-Beltrami operator on functions on $M$ is a generalized
Laplacian with trivial vector bundle $\VV$. In this special case,  the
limit of products of approximate 
kernels constructs the path
integral for ordinary (bosonic) quantum mechanics, and
Thm.~\ref{th:kgen}  shows this path integral computes the heat
kernel for the Laplace-Beltrami operator.
This implies results similar to those of Andersson and Driver
\cite{AD99}, though the convergence here is uniform rather than weak. 

The time-slicing approach to the path integral can
incorporate functions of paths. In fact, the given construction of the path
integral readily extends to
rigorously define the path integrals for $n$-point functions. The form
of the resulting expression
suggests the path integral of Prop.~\ref{pr:pathint} agrees with a
generalization of  Wiener
measure based on the heat kernel for generalized
Laplacians. 

It seems entirely plausible that Thm.~\ref{th:kgen} and
Prop.~\ref{pr:pathint} 
carry over to other sufficiently simple quantum theories. Indeed, this
was the authors' original motivation for constructing path integrals in
SUSYQM. For a first instance, the approach of one of the authors to
Yang-Mills on a Riemann surface~\cite{fine91} should easily combine
with the construction of the (bosonic) path integral given here to
provide a rigorous construction of the functional integral for the
expectation of certain classes of Wilson lines in that theory. 
In this and other cases,  the rigorous stationary phase argument
of Thm.~\ref{th:asit} may  apply, but with $\hbar$ as the
parameter. This would reproduce the semiclassical approximation or perhaps even the
full Feynman diagram expansion. This would be of particular interest
in cohomological field theories, where the stationary phase
approximation is exact. A rigorous
interpretation of the path integral, in which the semiclassical
approximation proves valid, would be the obvious starting point to
make rigorous several powerful path integral arguments in
cohomological field theories that yield interesting mathematical
results.

\bibliographystyle{alpha}

\begin{thebibliography}{Rog92b}

\bibitem[AD99]{AD99}
Lars Andersson and Bruce Driver.
\newblock Finite-dimensional approximations to {W}iener measure and path
  integral formulas on manifolds.
\newblock {\em J. Funct. Anal.}, 165(2):430--498, 1999.

\bibitem[AG83]{Alvarez83}
Luis Alvarez-Gaum\'e.
\newblock Supersymmetry and the {Atiyah-Singer} index theorem.
\newblock {\em Commun. Math. Phys.}, 90:161, 1983.

\bibitem[AJ90]{AJ90}
Michael~F. Atiyah and Lisa Jeffrey.
\newblock Topological {L}agrangians and cohomology.
\newblock {\em J. Geom. Phys.}, 7(1):119--136, 1990.

\bibitem[APS75]{APS75a}
Michael~F. Atiyah, Vijay~K. Patodi, and Isadore~M. Singer.
\newblock Spectral asymmetry and {R}iemannian geometry. {I}.
\newblock {\em Math. Proc. Cambridge Philos. Soc.}, 77:43--69, 1975.

\bibitem[APS76]{APS76}
Michael~F. Atiyah, Vijay~K. Patodi, and Isadore~M. Singer.
\newblock Spectral asymmetry and {R}iemannian geometry. {III}.
\newblock {\em Math. Proc. Cambridge Philos. Soc.}, 79(1):71--99, 1976.

\bibitem[Arn98]{Arnold98}
V.~I. Arnold.
\newblock {\em Ordinary Differential Equations}.
\newblock MIT Press, 1998.

\bibitem[AS53]{AS53}
W.~Ambrose and I.~M. Singer.
\newblock A theorem on holonomy.
\newblock {\em Trans. Amer. Math. Soc.}, 75:428--443, 1953.

\bibitem[Ati85]{Atiyah85}
Michael~F. Atiyah.
\newblock Circular symmetry and stationary-phase approximation. {Colloquium} in
  honor of {Laurent Schwartz}, vol. 1, ({Palaiseau}, 1983).
\newblock {\em Ast\'erisque}, 1(131):43--59, 1985.

\bibitem[BGV04]{BGV04}
Nicole Berline, Ezra Getzler, and Mich{\`e}le Vergne.
\newblock {\em Heat Kernels and {Dirac} Operators}.
\newblock Springer, 2004.

\bibitem[Bis84a]{Bismut84a}
Jean-Michel Bismut.
\newblock The {A}tiyah-{S}inger theorems: a probabilistic approach. {I}. {T}he
  index theorem.
\newblock {\em J. Funct. Anal.}, 57(1):56--99, 1984.

\bibitem[Bis84b]{Bismut84b}
Jean-Michel Bismut.
\newblock The {A}tiyah-{S}inger theorems: a probabilistic approach. {II}. {T}he
  {L}efschetz fixed point formulas.
\newblock {\em J. Funct. Anal.}, 57(3):329--348, 1984.

\bibitem[Bla93]{Blau93}
Matthais Blau.
\newblock The {M}athai-{Q}uillen formalism and topological field theory.
\newblock {\em J. Geom. Phys.}, 11(1-4):95--127, 1993.
\newblock Infinite-dimensional geometry in physics (Karpacz, 1992).

\bibitem[BP]{BP08}
Christian B\"ar and Frank Pf\"affle.
\newblock Path integrals on manifolds by finite dimensional approximations.
\newblock AP/07032731v1.

\bibitem[BT93]{BT93}
Matthias Blau and George Thompson.
\newblock {$N=2$} topological gauge theory, the {E}uler characteristic of
  moduli spaces, and the {C}asson invariant.
\newblock {\em Comm. Math. Phys.}, 152(1):41--71, 1993.

\bibitem[dC92]{doCarmo92}
Manfredo~Perdig{\~a}o do~Carmo.
\newblock {\em Riemannian geometry}.
\newblock Mathematics: Theory \& Applications. Birkh\"auser Boston Inc.,
  Boston, MA, 1992.
\newblock Translated from the second Portuguese edition by Francis Flaherty.

\bibitem[DH82]{DH82}
J.~J. Duistermaat and G.~J. Heckman.
\newblock On the variation in the cohomology of the symplectic form of the
  reduced phase space.
\newblock {\em Invent. Math.}, 69(2):259--268, 1982.

\bibitem[Eva98]{Evans98}
Lawrence~C. Evans.
\newblock {\em Partial Differential Equations}, volume~19 of {\em Graduate
  Studies in Mathematics}.
\newblock American Mathematical Society, 1998.

\bibitem[Fin91]{fine91}
Dana Fine.
\newblock Quantum {Yang-Mills} on a {Riemann} surface.
\newblock {\em Comm. Math. Phys.}, 140:321--338, 1991.

\bibitem[FPS75]{APS75b}
Atiyah~Michael F., Vijay~K. Patodi, and Isadore~M. Singer.
\newblock Spectral asymmetry and {R}iemannian geometry. {II}.
\newblock {\em Math. Proc. Cambridge Philos. Soc.}, 78(3):405--432, 1975.

\bibitem[FS08]{FS08}
Dana Fine and Stephen Sawin.
\newblock A rigorous path integral for supersymmetic quantum mechanics and the
  heat kernel.
\newblock {\em Comm. Math. Phys.}, 284(1):79--91, 2008.
\newblock arXiv:0705.0638.

\bibitem[FS14]{FS14}
Dana Fine and Stephen Sawin.
\newblock Short-time asymptotics of a rigorous path integral for n = 1
  supersymmetric quantum mechanics on a riemannian manifold.
\newblock {\em J. Math. Phys.}, 55(6), 2014.
\newblock arXiv:1207.2751.

\bibitem[FW84]{FW84}
Dan Friedan and Paul Windey.
\newblock Supersymmetric derivation of the {A}tiyah-{S}inger index and the
  chiral anomaly.
\newblock {\em Nuclear Phys. B}, 235(3):395--416, 1984.

\bibitem[Get86a]{Getzler86b}
Ezra Getzler.
\newblock The local {A}tiyah-{S}inger index theorem.
\newblock In {\em Ph\'enom\`enes critiques, syst\`emes al\'eatoires, th\'eories
  de jauge, {P}art {I}, {II} ({L}es {H}ouches, 1984)}, pages 967--974.
  North-Holland, Amsterdam, 1986.

\bibitem[Get86b]{Getzler86a}
Ezra Getzler.
\newblock A short proof of the local {A}tiyah-{S}inger index theorem.
\newblock {\em Topology}, 25(1):111--117, 1986.

\bibitem[Get91]{Getzler91}
Ezra Getzler.
\newblock The {T}hom class of {M}athai and {Q}uillen and probability theory.
\newblock In {\em Stochastic analysis and applications ({L}isbon, 1989)},
  volume~26 of {\em Progr. Probab.}, pages 111--122. Birkh\"auser Boston,
  Boston, MA, 1991.

\bibitem[Gil84]{gilkey84}
Peter~B. Gilkey.
\newblock {\em Invariance theory, the heat equation, and the {A}tiyah-{S}inger
  index theorem}, volume~11 of {\em Mathematics Lecture Series}.
\newblock Publish or Perish, Inc., Wilmington, DE, 1984.

\bibitem[MQ86]{MQ86}
Varghese Mathai and Daniel Quillen.
\newblock Superconnections, {Thom} classes, and equivariant characteristic
  classes.
\newblock {\em Topology}, 1986.

\bibitem[MS67]{MS67}
Henry~P. McKean, Jr. and Isadore~M. Singer.
\newblock Curvature and the eigenvalues of the {L}aplacian.
\newblock {\em J. Differential Geometry}, 1(1):43--69, 1967.

\bibitem[Pat71]{Patodi71}
V.~K. Patodi.
\newblock Curvature and the eigenforms of the {L}aplace operator.
\newblock {\em J. Differential Geometry}, 5:233--249, 1971.

\bibitem[Rog87]{Rogers87}
Alice Rogers.
\newblock A superspace path integral proof of the {G}auss-{B}onnet-{C}hern
  theorem.
\newblock {\em J. Geom. Phys.}, 4(4):417--437, 1987.

\bibitem[Rog92a]{Rogers92a}
Alice Rogers.
\newblock Stochastic calculus in superspace. {I}. {S}upersymmetric
  {H}amiltonians.
\newblock {\em J. Phys. A}, 25(2):447--468, 1992.

\bibitem[Rog92b]{Rogers92b}
Alice Rogers.
\newblock Stochastic calculus in superspace. {II}. {D}ifferential forms,
  supermanifolds and the {A}tiyah-{S}inger index theorem.
\newblock {\em J. Phys. A}, 25(22):6043--6062, 1992.

\bibitem[Rog03]{Rogers03}
Alice Rogers.
\newblock Supersymmetry and {B}rownian motion on supermanifolds.
\newblock {\em Infin. Dimens. Anal. Quantum Probab. Relat. Top.},
  6(suppl.):83--102, 2003.

\bibitem[See67]{seeley67}
R.~T. Seeley.
\newblock Complex powers of an elliptic operator.
\newblock In {\em Singular {I}ntegrals ({P}roc. {S}ympos. {P}ure {M}ath.,
  {C}hicago, {I}ll., 1966)}, pages 288--307. Amer. Math. Soc., Providence,
  R.I., 1967.

\bibitem[Wit82a]{Witten82a}
Edward Witten.
\newblock Constraints on supersymmetry breaking.
\newblock {\em Nuclear Phys. B}, 202(2):253--316, 1982.

\bibitem[Wit82b]{Witten82b}
Edward Witten.
\newblock Supersymmetry and morse theory.
\newblock {\em J. Differential Geom.}, 17(4):661--692 (1983), 1982.

\end{thebibliography}
\def\cprime{$'$}

\end{document}